\documentclass{article}

\usepackage{microtype}
\usepackage{graphicx}
\usepackage{subfigure}
\usepackage{booktabs} 

\usepackage{amsmath,amsfonts,bm}

\def\bx{\boldsymbol{x}}
\def\dd{\mathrm{d}}

\def\eqref#1{equation~\ref{#1}}

\def\1{\bm{1}}

\DeclareMathAlphabet{\mathsfit}{\encodingdefault}{\sfdefault}{m}{sl}
\SetMathAlphabet{\mathsfit}{bold}{\encodingdefault}{\sfdefault}{bx}{n}

\def\sS{{\mathbb{S}}}

\usepackage{url}
\usepackage{amsthm}
\usepackage{xcolor}
\usepackage{amssymb}
\usepackage{hyperref}

\newcounter{kaicomm}

\newcounter{bxincomm}
\definecolor{aqua}{rgb}{0.00,0.67,0.80}

\newtheorem{theorem}{Theorem}[section]

\newtheorem{definition}{Definition}[section]
\usepackage[accepted]{icml2022_TAGML}

\begin{document}
\twocolumn[
\icmltitle{Approximate Equivariance SO(3) Needlet Convolution}

\icmlsetsymbol{equal}{*}

\begin{icmlauthorlist}
\icmlauthor{Kai Yi}{equal,UNSW}
\icmlauthor{Jialin Chen}{equal,SJTU}
\icmlauthor{Yu Guang Wang}{UNSW,SJTU,SJTU-ZIAS}
\icmlauthor{Bingxin Zhou}{SJTU-ZIAS,USYD}\\

\icmlauthor{Pietro Li\`{o}}{Cam}
\icmlauthor{Yanan Fan}{UNSW}
\icmlauthor{Jan Hamann}{UNSW_physics}

\end{icmlauthorlist}

\icmlaffiliation{UNSW}{UNSW Data Science Hub, School of Mathematics and Statistics, University of New South Wales, Sydney, Australia}
\icmlaffiliation{SJTU}{Institute of Natural Sciences,
 School of Mathematical Sciences,
Shanghai Jiao Tong University,
Shanghai, China}
\icmlaffiliation{SJTU-ZIAS}{Zhangjiang Institute for Advanced Study, Shanghai Jiao Tong University, Shanghai, China}
\icmlaffiliation{USYD}{The University of Sydney Business School, The University of Sydney, Sydney, Australia.}
\icmlaffiliation{Cam}{Department of Computer Science and Technology, University of Cambridge,
Cambridge, United Kingdom}
\icmlaffiliation{UNSW_physics}{School of Physics, University of New South Wales, Sydney, Australia}

\icmlcorrespondingauthor{Yu Guang Wang}{yuguang.wang@sjtu.edu.cn}

\icmlkeywords{Machine Learning, ICML}

\vskip 0.15in
]
\printAffiliationsAndNotice{\icmlEqualContribution}

\begin{abstract}
This paper develops a rotation-invariant needlet convolution for rotation group SO(3) to distill multiscale information of spherical signals.
The spherical needlet transform is generalized from $\sS^2$ onto the SO(3) group, which decomposes a spherical signal to approximate and detailed spectral coefficients by a set of tight framelet operators. The spherical signal during the decomposition and reconstruction achieves rotation invariance. 
Based on needlet transforms, we form a Needlet approximate Equivariance Spherical CNN (NES) with multiple SO(3) needlet convolutional layers. The network establishes a powerful tool to extract geometric-invariant features of spherical signals. 
The model allows sufficient network scalability with multi-resolution representation. A robust signal embedding is learned with wavelet shrinkage activation function, which filters out redundant high-pass representation while maintaining approximate rotation invariance. 
The NES achieves state-of-the-art performance for quantum chemistry regression and Cosmic Microwave Background (CMB) delensing reconstruction, which shows great potential for solving scientific challenges with high-resolution and multi-scale spherical signal representation.
\end{abstract}


\section{Introduction}
\begin{figure*}[th]
    \centering
    \includegraphics[width=\textwidth]{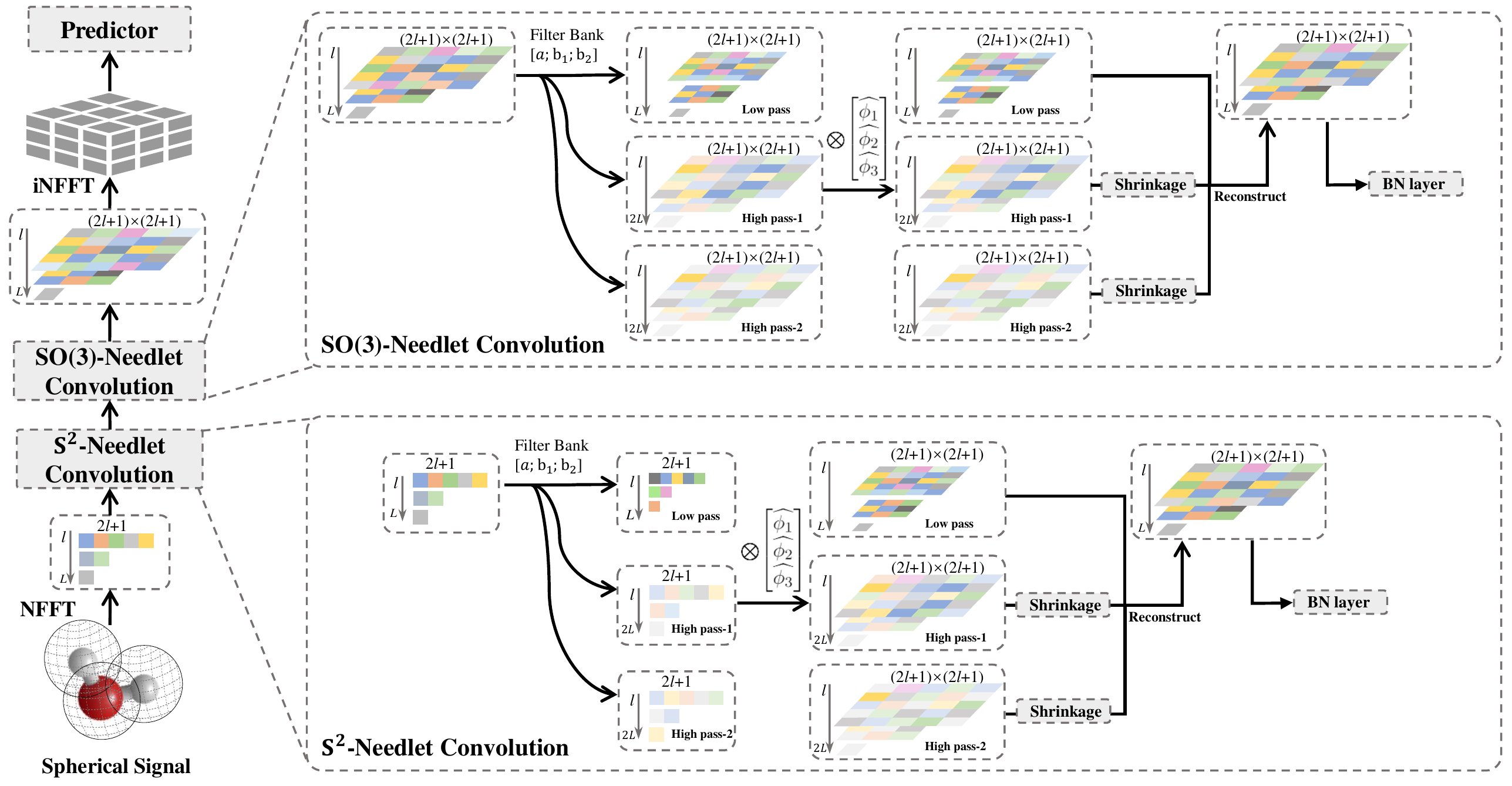}
    \vspace{-6mm}
    \caption{This figure shows the framework of our NES. As the left column shows, we first carry on a non-equispaced fast Fourier transform (NFFT) with predefined weights on the spherical signal. The following are an $\mathbb{S}^2$-Needlet Convolution and $\mathrm{SO(3)}$-Needlet Convolutions, which can be used to decompose the signal in multi scales. Then, we use the inverse NFFT (iNFFT) over the output of the $\mathrm{SO(3)}$-Needlet Convolution and feed the reconstructed signal into the downstream predictor.}
    \label{fig:needlet}
\end{figure*}

Many data types in the real world can be modeled as spherical data, such as omnidirectional images \cite{coors2018spherenet}, molecules \cite{boomsma2017spherical}, and cosmic microwave background \cite{akrami2020planck}. Such spherical signals contain abundant topological features. Unfortunately, existing research \cite{caldeira2019deepcmb, yi2020cosmo} usually maps spherical signals to $\mathbb{R}^2$ for convenient modeling with convolutional neural networks (CNNs), which results in distorted signals and ineffective shift equivariance \cite{marinucci2008spherical}. 

Alternatively, geometric deep learning \cite{bronstein2017geometric,bronstein2021geometric} provides a universal blueprint for learning stable representation of high-dimensional data in different domains to build equivariant or invariant neural network layers that respect exact or approximate data symmetries, such as translation, rotation, and permutation. As a fundamental requirement for many applications, 
it has been proven critical to preserving the symmetry property in deep learning algorithm design \cite{baek2021accurate,davies2021advancing,mendez2021geometric}.

Equivariance is a significant property of geometric deep learning models as required by many physical sciences, such as chemistry \cite{atz2021geometric} and biology \cite{townshend2021geometric}. This paper develops a scalable geometric deep learning model for spherical signal processing and learning with theoretically guaranteed rotation equivariance. Our model is based on needlet convolution on $\mathbb{S}^2$ and rotation group SO(3). The former describes the data representation on spherical point locations, while the latter records three-dimensional rotation angles of the signal. The input data features are embedded in each spherical point.
The main convolution computational unit is based on spherical needlets, which define a wavelet-like system on the two-dimensional sphere $\mathbb{S}^{2}$ that forms a tight frame on the sphere \cite{narcowich2006decomposition,narcowich2006localized,wang2017fully}. A needlet is characterized by a highly-localized spherical radial polynomial, which covers a large scale but captures detailed features in local regions. 

The \emph{needlet convolution} on SO(3) decomposes spherical signals into low-pass and high-pass needlet coefficients. By separately storing and processing approximate and detailed information of the input, the network establishes hidden embeddings with enhanced scalability. In addition, the wavelet shrinkage operation \cite{donoho1995noising,baldi2009asymptotics} gains robust representations by filtering out redundant high-pass information in the framelet domain.
The exact multiscale embeddings by SO(3)-needlet convolutions are invariant to rotation. Such convolutions can construct a deep neural network that distills the geometric invariant features of a spherical signal. We name it \textbf{N}eedlet approximate \textbf{E}quivariance \textbf{S}pherical CNN (\textbf{NES}). Inside the network, we utilize the convolution over the rotation group in multi scales to guarantee rotation equivariance. 

The NES with shrinkage activation gains provably approximate equivariance, where the equivariance error diminishes at sufficiently high scales. Moreover, the needlet convolution is implemented efficiently with fast Fourier transforms (FFTs) on the sphere and rotation group. We validate the proposed NES on different real-world scientific problems with high-resolution and multi-scale spherical signal inputs including regressing quantum chemistry molecules and 
reconstructing lensing Cosmic Microwave Background, for which our method achieves state-of-the-art performance.

\section{Spherical Needlet Framework}
\label{sec:method}
Needlets are a type of framelets \cite{wang2020tight,han2017framelets} that enjoys good localization properties in both spatial and harmonic spaces. We formulate a spherical needlet transform, which projects the given spherical signal to a set of multi-scale needlet representations in the framelet domain. The new representations can be uniquely decomposed. They are easy to compute, and divide approximate and detailed information into different scale levels, as traditional wavelets.

\subsection{Characterization of Multi-scale Spherical Needlets}
\label{MRA}
Needlets are defined on a Riemannian manifold $\mathcal{M}$. This paper considers a special case of $\mathcal{M}$, i.e., on $\mathbb{S}^2$ or $\mathrm{SO(3)}$. 
We define the spherical needlets with a \emph{filter bank} $\eta:=\left\{a ; b^{1}, \ldots, b^{r}\right\} \subset l_{1}(\mathbb{Z}):=\left\{h=\left\{h_{k}\right\}_{k \in \mathbb{Z}} \subset \mathbb{C}: \sum_{k \in \mathbb{Z}}\left|h_{k}\right|<\infty\right\}$ and a set of associated \emph{generating functions} $\Psi=\{\alpha;\beta^1,\cdots,\beta^r\} \subset L_1(\mathbb{R})$. 
We name filter $a$ the \emph{low-pass filter}, and filters $\{b^1,\cdots,b^r\}$ the \emph{high-pass filters}. The former distills approximate information from the input signal, and the latter reserves more detailed information and together with noise. 
The associated generating functions and filter bank satisfy the relationship
\begin{equation}
    \label{relations}
    \widehat{\alpha}(2\xi)=\widehat{a}(\xi) \widehat{\alpha}(\xi),\quad
    \widehat{\beta^{n}}(2 \xi)=\widehat{b^{n}}(\xi) \widehat{\alpha}(\xi),
\end{equation}
where $n=1, \ldots, r,$ and $\xi \in \mathbb{R}$.

To discretize the continuous needlets with zero numerical error, we utilize \textit{Polynomial-exact Quadrature Rule} \cite{wang2017fully} that are generated by the tensor product of the Gauss-Legendre nodes on the interval $[-1,1]$ and equi-spaced nodes in the dimension with non-equal weights, such as longitude on sphere. Let $v_{j,k}$ represent low-pass coefficients, and $w^n_{j,k}$ represent high-pass coefficients of the signal function $f$, where $k=0, \ldots, N_{j+1}$  and $j \geq J$, $N_j$ is the number of sampling points at scale $j$. The low-pass (or high-pass) coefficients are defined by the inner products of low-pass (or high-pass) needlets and $f$. In practice, we calculate the coefficients in the Fourier domain for fast computation by
\begin{equation}
\label{freq_vw_define_1} 
    \widehat{v}_{j, \ell}=\widehat{f}_{\ell} \overline{\widehat{\alpha}\left(\frac{\lambda_{\ell}}{2^{j}}\right)},\quad \widehat{w}_{j-1, \ell}^{n}=\widehat{f}_{\ell} \overline{\widehat{\beta^{n}}\left(\frac{\lambda_{\ell}}{2^{j-1}}\right)}.
\end{equation}
We denote $\widehat{f}_{\ell}$ as the generalized Fourier coefficients of $f$ at degree $\ell$.
More details about the filter bank and construction of needlets on $\mathbb{S}^2$ and $\mathrm{SO(3)}$ are given in Appendix~\ref{needlet}.

\subsection{Spherical Needlet Convolution}
\label{paragraph_convolution}
The \emph{spherical needlet convolution} on $\mathcal{M}$ is defined by
\begin{equation}
\label{eq:NESconv_definition}
    [\phi\star f](R) = \langle L_R\phi,f\rangle= \int_{\mathcal{M}}\phi(R^{-1}x)f(x)dx,
\end{equation}
where $f$ is a signal, $\phi$ is a learnable locally supported filter, $L_R\phi(x) = \phi(R^{-1}x)$, and $\mathcal{M}$ represents $\mathbb{S}^2$ or $\mathrm{SO(3)}$. 
The constructed needlet convolution is rotation equivariant. Formally, a neural network (i.e., a function on $\mathcal{M}$) $\Phi$ is said \emph{rotation equivariant} if for an arbitrary rotation operator $L_R$, there exists an operator $T_R$ such that
$\Phi \circ L_R = T_R\circ \Phi$.
A rotation equivariant neural network provides more efficient and accurate prediction with theoretical support, which properties are desired for rotatable signals.
It is provable that the convolution in (\ref{eq:NESconv_definition}) satisfies the Fourier theorem, i.e., $[\widehat{\phi\star f}]_\ell = \widehat{f}_\ell\cdot\widehat{\phi}_{\ell}^\dagger$, where $\dagger$ denotes the conjugate transpose and $\ell$ is the degree parameter. The operation $\cdot$ is matrix multiplication for $\mathrm{SO(3)}$ and outer product for $\mathbb{S}^2$.

The formulation in (\ref{eq:NESconv_definition}) has been adopted by  
Spherical CNN \citep{cohen2018spherical}, which induces convolution on Fourier coefficients. We define the convolution using needlet coefficients of a spherical signal.
We construct the needlet coefficients with the needlet system defined in Section~\ref{MRA}. We take $n=2$ and get $\{\widehat{v}_{1,\ell}\}_{\ell=1}^{\Lambda_{J_0}}$, $\{\widehat{w}_{1,\ell}^{1}\}_{\ell=1}^{\Lambda_{J_1}}$ and $\{\widehat{w}_{1,\ell}^{2}\}_{\ell=1}^{\Lambda_{J_1}}$ for a low-pass and two high-pass needlet coefficients, where $\Lambda_j$ denotes sequence length of Fourier series of quadrature rule sampling points at scale $j$, and $J_0, J_1$ are the scale of low pass and high pass respectively.

These needlet coefficients can be used to reconstruct the Fourier coefficients $\hat{f}$ of signal $f$ of degree $\ell$.
We denote this relation as $\begin{bmatrix}\widehat{v}_{1,\ell},\widehat{w}_{1,\ell}^{1}, \widehat{w}_{1,\ell}^{2}\end{bmatrix}^\top \asymp \widehat{f}_\ell$, where $\asymp$ means formal equivalence. We hereby establish formally an equivalent expression of $[\widehat{\phi\star f}]_\ell$ with multi-scale information and rotation equivariance:
\begin{align*}
\begin{bmatrix}
    \widehat{\phi_1}_\ell\\\widehat{\phi_2}_\ell\\\widehat{\phi_3}_\ell
\end{bmatrix}
\odot\widehat{f}_{\ell} &\asymp
\begin{bmatrix}
    \widehat{\phi_1}_\ell\\\widehat{\phi_2}_\ell\\\widehat{\phi_3}_\ell
\end{bmatrix}
\odot 
\begin{bmatrix}
    \widehat{v}_{1,\ell}\\\widehat{w}_{1,\ell}^{1}\\ \widehat{w}_{1,\ell}^{2}
\end{bmatrix}
= 
\begin{bmatrix}
    \widehat{\phi_1}_\ell\cdot\widehat{f}_\ell\overline{\widehat{\alpha}\left(\frac{\lambda_{\ell}}{2^{J_0}}\right)}\\
    \widehat{\phi_2}_\ell\cdot\widehat{f}_\ell\overline{\widehat{\beta^1}\left(\frac{\lambda_{\ell}}{2^{J_0}}\right)}\\
    \widehat{\phi}_{3,\ell}\cdot\widehat{f}_\ell\overline{\widehat{\beta^2}\left(\frac{\lambda_{\ell}}{2^{J_0}}\right)}
\end{bmatrix} \\[1mm]
&= 
\begin{bmatrix}
    [\widehat{\phi_1\star f}]_\ell\overline{\widehat{\alpha}\left(\frac{\lambda_{\ell}}{2^{J_0}}\right)}\\
    [\widehat{\phi_2\star f}]_\ell\overline{\widehat{\beta^1}\left(\frac{\lambda_{\ell}}{2^{J_0}}\right)}\\
    [\widehat{\phi_3\star f}]_\ell\overline{\widehat{\beta^2}\left(\frac{\lambda_{\ell}}{2^{J_0}}\right)}
\end{bmatrix} 
\asymp[\widehat{\phi\star f}]_\ell.
\end{align*}
Here $\widehat{\phi_i}_\ell$ $(i=1,2,3)$ are three learnable filters defined in the frequency domain and $\odot$ is the Hadamard product.

\subsection{Rotation equivariance Error}
\label{paragraph_shrinkage}
\paragraph{Shrinkage Function}
One potential drawback of the spherical CNNs comes from the non-linear activation in each layer. The Fourier transforms introduce redundancy to feature representation in the frequency domain, which results in a heavy computational cost. To best preserve rotation equivariance at a reduced computational complexity, we employ a non-linear activation directly in the frequency domain with a small rotation equivariance error. Similar to UFGConv~\citep{zheng2021framelets}, we cut off the high-pass coefficients $x$ in the frequency domain by a shrinkage thresholding, i.e., 
\begin{equation*}
    \textup{Shrinkage}(x)=\textup{sgn}(x)(|x|-\lambda)_{+}, \forall x\in \mathbb{R},
\end{equation*}
with the threshold value $\lambda = \sigma\sqrt{2\textup{log}(N)}/\sqrt{N}$ for $N$ coefficients. The hyperparameter $\sigma$ is an analogue to the noise level of the denoising model. Note that we do not cut off the low-pass framelet coefficients, as they distill important approximate information of input data. It is critical to offering an approximate rotation equivariance for the shrinkaged needlet convolution, as we discuss as follows.
\begin{theorem}
\label{error_thm}
Let $\mathbb{W}^s_p(\mathbb{S}^2)$ with $s> 2/p$ and $1\leq p \leq \infty$ be a Sobelev space embedded in $\mathbb{L}_p(\mathbb{S}^2)$. For $f \in \mathbb{W}^s_p(\mathbb{S}^2)$, $\phi$ is a filter, then the rotation equivariance error due to using the shrinkage function is defined as the maximum of the following over all $R\in\mathrm{SO(3)}$,
\begin{align*}
    &\mathcal{E}(f)=\\
    &\max_{R\in\mathrm{SO(3)}}\sum_{\ell=0}^B\left\| \textup{Shr}(\widehat{L_Rf\star \phi})_\ell^{(\textup{H})}-D^\ell(R)\textup{Shr}(\widehat{f\star\phi}_\ell^{(\textup{H})})\right\|^2,
\end{align*}
where $B$ is the bandwidth for spherical signal embedding depending on the specific quadrature rule used, \textup{Shr($\cdot$)} represents shrinkage function, superscript $(\textup{H})$ indicates the high-pass coefficients. Then, the approximate equivariance error for $f$ is
\begin{equation*}
    \mathcal{E}(f)\leq C 2^{-(J_0+1)s},
\end{equation*}
where $J_0$ is the scale of the low pass, $C$ is a constant depending only on $s$, $\phi$ and the Sobolev norm of $f$.
\end{theorem}
The shrinkage mechanism thus introduces a stable rotation equivariance error. 
The condition in Theorem~\ref{error_thm} $s>2/p$ indicates that each function of $\mathbb{W}^s_p(\mathbb{S}^2)$ has a representation in the continuous function space on $\mathbb{S}^2$. Then, the numerical computation for $f$ makes sense.
The proof of Theorem~\ref{error_thm} is given in Appendix~\ref{bound}.

\paragraph{Pooling Operator} 
We also establish a spectral pooling in the frequency domain to circumvent repeated Fourier transforms. A spectral pooling removes coefficients with degree larger than $\ell/2$ for the spectral representation $\hat{f} = [\hat{f}_{0},\hat{f}_{1},\cdots,\hat{f}_{\ell}]$. We prove that the spectral pooling operator is rotation equivariant (see Appendix~\ref{Pooling_proof}). 

\paragraph{Network Architecture}
The framework is scalable to any application scenarios that can be represented by spherical signals. We illustrate the overall workflow of our proposed Needlet Spherical CNNs in Figure~\ref{fig:needlet} with application scenario for bio-molecular prediction, where the input is a set of spherical signals centered at atoms of a molecule. 
The spherical data is sampled on the points of a polynomial-exact quadrature rule. Based on the rule, we implement non-equispaced fast Fourier transforms (NFFTs) with predefined weights. The Fourier representations are sent through an $\mathbb{S}^2$-needlet convolution to $\mathrm{SO(3)}$. A number of rotation equivariant $\mathrm{SO(3)}$-needlet convolutions are repeatedly conducted. Inside each needlet convolution, we use a wavelet shrinkage to threshold small high-pass coefficients, following a pooling operator to downsize the representation. The final output of $\mathrm{SO(3)}$-needlet convolution is handled by the inverse NFFT (iNFFT) to feed into a downstream predictor.

\section{Experiments}
The main advantage of our model is the property of equivariance to $\mathrm{SO(3)}$ transforms with multiscale representation for complicated real-world application. This section validates the model with three experiments.
Our models are trained on 24G NVIDIA GeForce RTX 3090 Ti GPUs. The hyperparameters are obtained by grid search. Adam \cite{kingma2014adam} is used as our optimizer.

\subsection{Local MNIST Classification}
\begin{figure}[t]
    \centering
    \includegraphics[width = 0.45\textwidth]{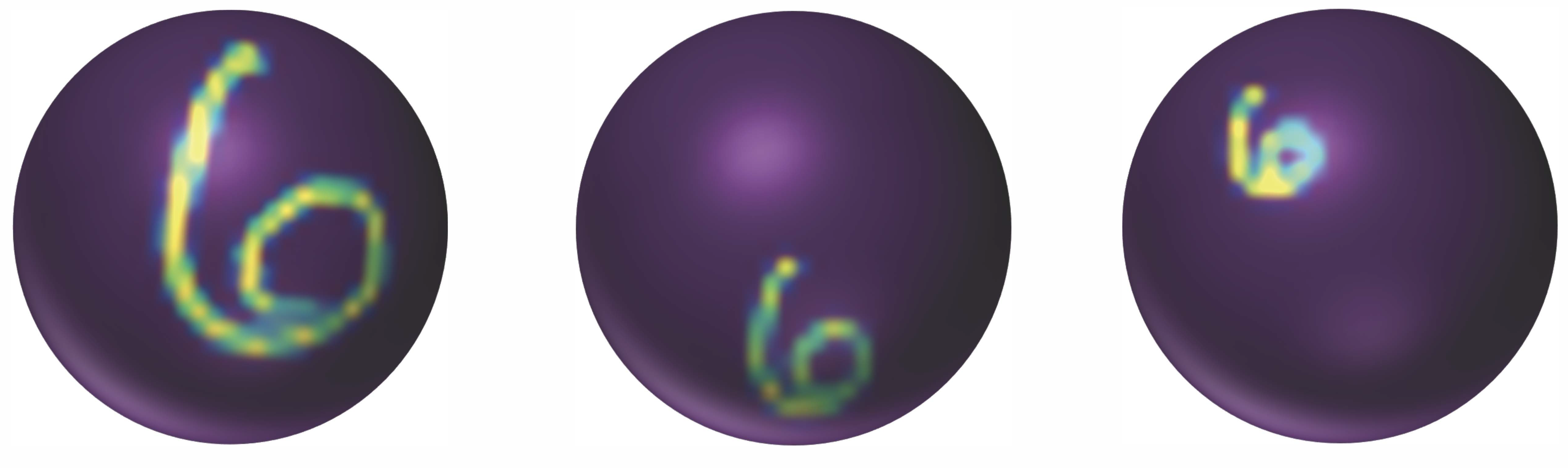}
    \caption{Illustration of a projected MNIST digit onto the sphere with $3$ different downscale ratios ($10\%$, $50\%$, $90\%$, left to right). The higher downscale indicates that the size of the digit is smaller on the sphere, which increases the difficulty of the model feature extraction where more detailed information needs captured.}
    \label{Local Mnist}
\end{figure}

\begin{table}
    \centering
    \setlength{\tabcolsep}{1.6mm}
    \caption{Test accuracy on spherical MNIST with varying scales. 
    }
    \vspace{1mm}
    \begin{tabular}{l|rrrrr}
    \toprule
    Downscale Ratio & 10\%   & 30\%   & 50\%   & 70\%   & 90\%   \\ 
    \midrule
    Spherical CNN & 94.99 & 92.17 & 86.92 & 83.73 & 78.71 \\
    NES & \textbf{97.84} & \textbf{97.30} & \textbf{96.74}&  \textbf{95.21} & \textbf{92.66}\\ 
    \bottomrule
    \end{tabular}
    \label{mnist_performance}
\end{table}

\paragraph{Dataset}
The first experiment evaluates the effectiveness of the needlet convolution neural network in capturing high-frequency information. We follow \citet{cohen2018spherical} and use a modified spherical MNIST classification dataset, where the images are projected onto a sphere to establish rotated training and test sets. Here the samples of the training set are all rotated by the same rotation while those in the test set are rotated by another rotation. We downscale the original MNIST images into five different resolutions and then project them onto a scalable area of the sphere. 

\paragraph{Setup}
Our model is compared with Spherical CNN \cite{cohen2018spherical}. We adopt the same architecture $\mathbb{S}^2$conv-ReLU-$\mathrm{SO(3)}$conv-ReLU-FC-softmax, with bandwidth $L$ = 30, 10, 6 and $k$ = 20, 40, 10 channels: when it comes to our model, we replace $\mathbb{S}^2$conv and $\mathrm{SO(3)}$conv with $\mathbb{S}^2$-needlet convolution and $\mathrm{SO(3)}$-needlet convolution, bandwidth $L$ = 30, 10, 6 and $k$ = 20, 40, 10 channels. We select the batch size of $64$ and learning rate  $0.001$.

\paragraph{Results}
The test accuracy for spherical MNIST is presented in Table~\ref{mnist_performance}. 
To test the rotation equivariance of the models, we rotate the training dataset and test dataset with two different rotations. That is, the input training data are all rotated with a same rotation in SO(3), and all test data are rotated by another rotation. We also test on downscaled datasets with various scales: the higher the scale, the less size of the spherical digit is on the sphere, and the signal is more localized.
Table~\ref{mnist_performance} indicates that both models keep high test accuracy with both training and test data rotated. We can observe that our model consistently achieves high accuracy on datasets for different downscale ratios. Especially for the high ratio, the digit is concentrated at a small region, and the model is required to capture more details of the spherical data.
In contrast, Spherical CNN has poorer performance with higher downscale. It demonstrates a reliable performance of our model in effectively distilling detailed and local features while maintaining rotation equivariance of the needlet convolutional layer.

\subsection{Molecular Property Prediction}
\begin{figure}[t]
    \centering
    \includegraphics[width=0.25\textwidth]{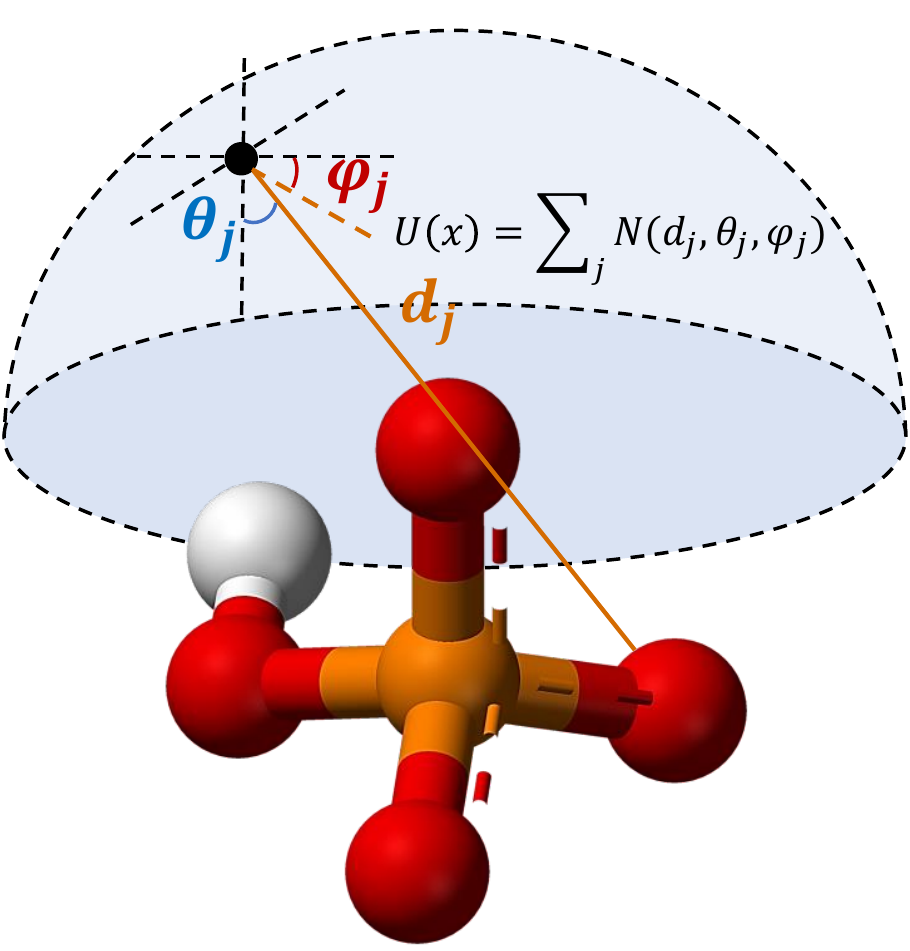}
    \vspace{-2mm}
    \caption{An illustration for computing the spherical signal of a molecule. We aggregate the information of each atom in the molecule with relative distance, polar angle and azimuthal angle.}
    \label{fig:md_sphere}
    \vspace{-0.2cm}
\end{figure}

\begin{table}[t]
\centering
\caption{Test RMSE of atomic energy on \textbf{QM7}. 
$\dagger$ indicates the method is rotation equivariant.}
\label{qm7_results}
\vspace{2mm}
\resizebox{0.85\linewidth}{!}{
\begin{tabular}{lcc}
\toprule
Method & RMSE &Params\\ \midrule
\textsc{MLP/Sorted CM} & 16.06 & - \\
\textsc{MLP/Random CM} & 5.96 & - \\
\textsc{GCN} & 7.32 $\pm$ 0.23&0.8M\\
\textsc{Spherical CNN}$^{\dagger}$   & 8.47&1.4M\\
\textsc{Clebsch–Gordan}$^{\dagger}$  & 7.97& $\geq$1.1M \\
\midrule
\textsc{NES}$^{\dagger}$ (Ours)&\textbf{7.21} $\pm$ 0.46  &0.9M \\\bottomrule
\end{tabular}
}
\end{table}

\begin{table}[ht]
    \centering
    \setlength{\tabcolsep}{0.31mm}
    \caption{Test MAE of forces in $\mathrm{meV}/\AA$ on \textbf{MD17}. 
    }
    \label{md17_performance}
    \vspace{1mm}
    \resizebox{0.9\linewidth}{!}{
    \begin{tabular}{lccccc}
    \toprule
         Molecule&sGDML&SchNet&DimeNet&SphereNet&NES\\
    \midrule
         Aspirin&29.5&58.5&21.6&18.6&\textbf{15.2}\\
         Ethanol&14.3&16.9&10.0&\textbf{9.0}&9.2\\
         Malonaldehyde&17.8&28.6&16.6&14.7&\textbf{13.6}\\
         Naphthalene&4.8&25.2&9.3&7.7&\textbf{3.5}\\
         Salicylic&\textbf{12.1}&36.9&16.2&15.6&14.2\\
         Toluene&\textbf{6.1}&24.7&9.4&6.7&\textbf{6.1}\\
         Uracil&\textbf{10.4}&24.3&13.1&11.6&10.8\\
    \bottomrule
    \end{tabular}
    }
\end{table}

\paragraph{Datasets}
The second experiment predicts molecular property over two widely used datasets (\textbf{QM7} and \textbf{MD17} \cite{chmiela2017machine} ) to evaluate the model’s expressivity to bio-molecular simulation.
\textbf{QM7} contains $7,165$ molecules. Each molecule contains at most $N=23$ atoms of $T=5$ types (H, C, N, O, S), which is to regress over the atomic energy of molecules given the corresponding position $p_i$ and charges $z_i$ of each atom $i$. 
\textbf{MD17} predicts the energies and forces at the atomic level for several organic molecules with up to $21$ atoms and four chemical elements, using the molecular dynamics trajectories.

We follow \citet{rupp2012fast} to generate spherical signals for every molecule. We define a sphere $S_i$ centered at $p_i$ for each atom $i$ and define the potential functions as 
$U_z(x) = \sum_{j\neq i,z_j=z}\frac{z_i\cdot z}{\|x-p_i\|}$, 
where $z$ is the charge of the atom, and $x$ is taken from $\mathbb{S}^2$. For every molecule, $N$ spherical signals are produced in $T$ channels. We use the Gauss-Legendre rule to discretize the continuous functions on the sphere with $L=20$ and create a sparse $N\times T\times (2L+1)\times(L+1)$ tensor as the input signal representation. 
For \textbf{QM7}, we generalize the \textit{Coulomb matrix} ($C\in \mathbb{R}^{N\times N}$) and obtain $23$ spherical signals for every molecule. 
For \textbf{MD17}, we create $N$ spherical signals that are centered at the positions of each atom for every sample, where $N$ is the number of atoms in the molecule. For the atom $i$, we define a corresponding spherical signal $U_i(x)$, where $x$ is taken from the sphere by the Gauss-Legendre sampling method. The relative position of each atom to $x$ is calculated with the absolute Cartesian coordinates of atoms provided by \textbf{MD17}. The spherical signal $U_i$ is defined as 
    $U_i(x) = \sum_{j=1}^{N}\mathcal{N}(d_j,\theta_j,\varphi_j)$,
where $(d_j,\theta_j,\varphi_j)$ is the position of atom $j$ relative to $x$. The $d_j$, $\theta_j$, and $\varphi_j$ denote the radial distance, polar angle, and the azimuthal angle respectively (see Figure~\ref{fig:md_sphere}). Different with \textbf{QM7}, \textbf{MD17} does not have a \textit{Coulomb matrix}. The number of spherical signals $N$ can thus be taken from a neural network or a mathematical operator to extract effective features with the relative positions. Here we choose the first approach of neural networks to adaptively learn feature. We fine-tune the hyperparameters individually for every type of molecules on the validation sets with $1,000$ samples for each type.

\paragraph{Setup} 
The bandwidth $L$ is from $20, 20, 10, 10, 5$ to $5$ in the final block and the feature dimension is from $5, 5, 8, 16, 32$ to $64$. The hyperparameter $\sigma$ is taken as $0.001$ for shrinkage. We run $10$ epochs for \textbf{QM7} with a batch size of $32$ and a learning rate of $5e-4$. For \textbf{MD17}, we choose a batch size of $32$ and a learning rate of $2e-4$. We run the model for $1,000$ epochs.

\paragraph{Results} 
We report the experimental results of \textbf{QM7} and \textbf{MD17} respectively in Tables \ref{qm7_results}-\ref{md17_performance}. 
For \textbf{QM7}, we compare the root mean square error (RMSE) of our proposed NES with MLP/Random CM, MLP/Sorted CM \cite{montavon2012learning}, GCN \cite{kipf2016semi}, Spherical CNN \cite{cohen2018spherical} and Clebsch-Gordan Net \cite{kondor2018clebsch}. The scores are averaged over $10$ trials with standard deviation. Our model uses approximately $0.9$ million parameters to achieve the lowest RMSE at $7.21 \pm 0.46$ among all rotation equivariant models. Our model enjoys the advantages of both a smaller number of parameters and a lower prediction error, owing to the incorporation of efficient computation and multiscale analysis architecture. 
In \textbf{MD17} task, we focus on atomic forces and measure the mean absolute error (MAE) averaged over all samples and atoms. 
SchNet \cite{schutt2017schnet} and DimeNet \cite{chmiela2018towards} are 3D graph models that incorporate relative distance information. SphereNet \cite{liu2021spherical} is a 3D graph model with physically-based representations of geometric information. Most of previous state-of-the-art models are graph-based models with hand-engineered features or expert knowledge. Instead, our model utilizes the adaptive learning of input features and incorporates multiscale analysis to improve the representation ability. Results show that the proposed model outperforms baseline models with strong performance and better generalization in molecular simulation, due to the rotation equivariance. NES achieves better performance on four types of molecules. Compared to NES, sGDML \cite{chmiela2018towards} is a kind of kernel method, which relies on human expertise and extra annotation, thus suffering from poor generalization to a new type of molecule.

\subsection{Delensing Cosmic Microwave Background}
\begin{figure}
    \centering
    \includegraphics[width=0.85\linewidth]{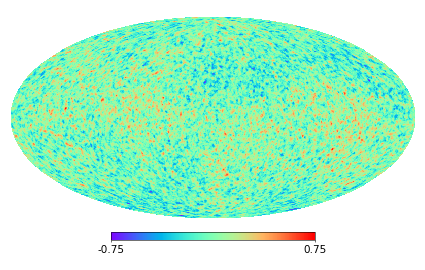}
    \vspace{-4mm}
    \caption{B-CMB multipoles unlensed map with tensor-to-scalar ratio $r =0.2$, which is one of the main constraints in detecting the Primordial Gravitational Wave Background. We color the map with the intensity values to predict.}
    \label{B_mode}
\end{figure}

The existence of a stochastic Primordial Gravitational Wave Background (PGWB) is a common prediction in the majority of inflationary models. It is formed when microscopic quantum fluctuations of the metric were stretched up to super-horizon scales by the sudden expansion of space-time that occurred during inflation \cite{caprini2018cosmological}. Since it has been able to free-stream from time as early as (possibly) the Planck time, PGWB has the potential of becoming one of the most powerful cosmological probes. The information about phase transitions and particle creation/annihilation may have taken place in the early universe, which allows new independent measurements of cosmological parameters. In order to discover PGWB, we need to constrain some parameters, such as the ratio between tensor and scalar perturbations $r=\mathcal{P}_t(k)/\mathcal{P}_s(k)$. Such a parameter relies on a high signal-to-noise ratio (SNR) reconstruction of the lensing potential, i.e., the projected weighted gravitational potential along the line-of-sight between us and the CMB. Photons in the CMB are deflected by the intervening mass distributions when they travel to us. The lensing effect distorts the recombination of the CMB and interferes with our ability to constrain early universe physics. Therefore, removing the lensing effect from observed data is critical to decoding early-universe physics. In this experiment, we use NES convolution to reconstruct the unlensed B-mode (Figure~\ref{B_mode}) component of the CMB polarization from the lensed $Q, U$ maps that are orthonormal bases corresponding to Stokes parameters.

\paragraph{Dataset}
Spherical CNN and NES are used to reconstruct the unlensed $B$ map from the lensed $Q, U$ maps. We simulate $10,000$ lensed $Q, U$ maps and B-CMB multipoles unlensed map with $N_{\rm Side}=64$. Then, transforming the original sample rules from HEALPix to Gauss-Legendre tensor product rules with the bandwidth $L = 128$ by taking the average of the four nearest HEALPix points in Gauss-Legendre coordinates. We split the whole dataset into 80\% training, 10\% validation, and 10\% test sets.

\paragraph{Setup}
We follow the U-Net architecture from \cite{caldeira2019deepcmb} and replace the standard image convolution with Spherical CNN and NES. In the encoder, the bandwidth is $128, 64, 32, 16$ and the feature dimensions are $2, 16, 32, 64$ respectively in each block. In the decoder, the bandwidth increases from $16, 32, 64$ to $128$. The encoder layers are skip-connected to decoder layers, which is consistent with the standard architecture of U-Net. We sum the mean squared error (MSE) in the pixel domain and in power spectrum as the loss function. We use a batch size of $16$, a learning rate of $5e-5$, and weight decay of $3e-4$ to train the model.

\begin{figure}[t]
    \centering
    \includegraphics[width=0.85\linewidth
    ]{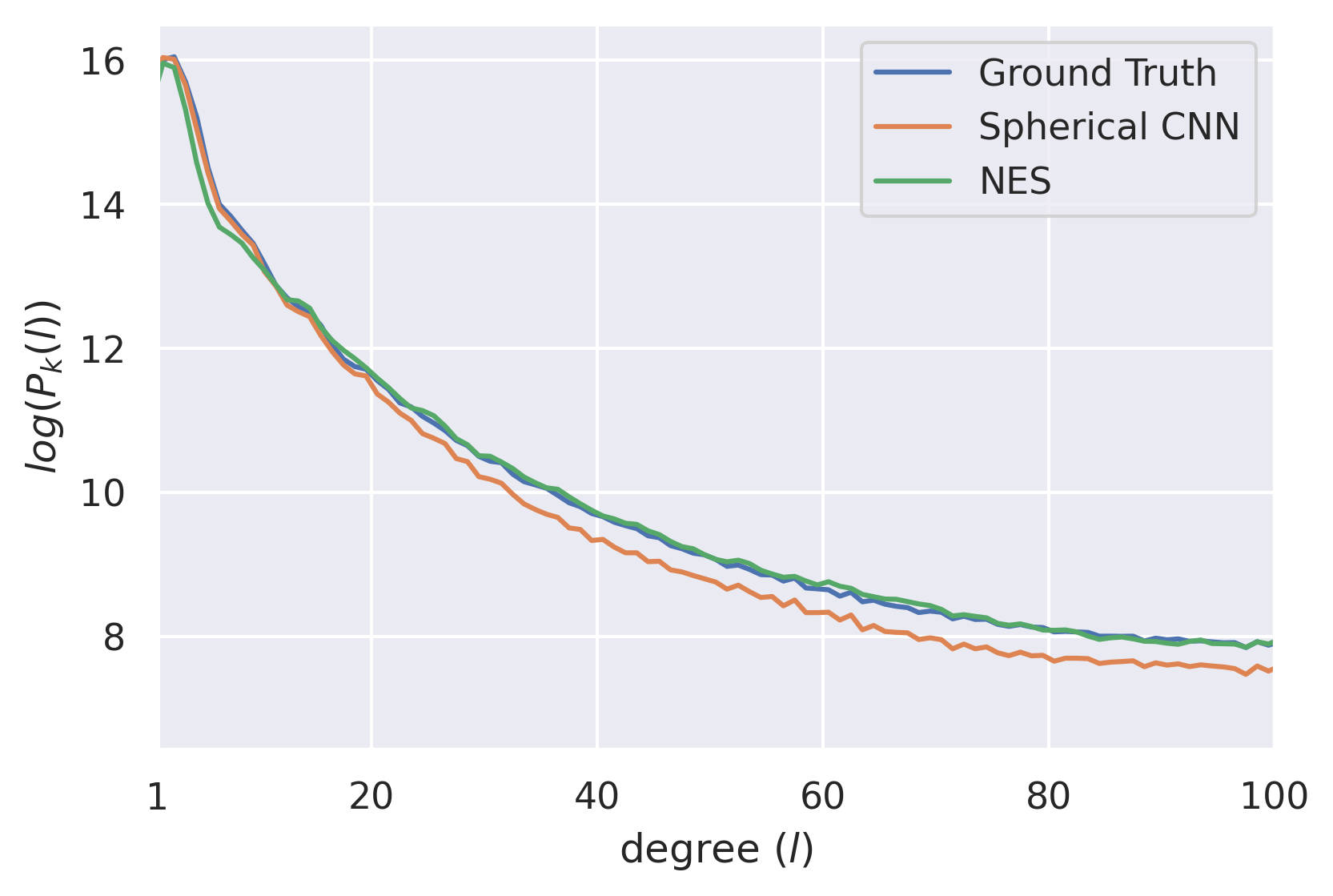}
    \vspace{-2mm}
    \caption{The power spectrum of unlensed B-map predicted by NES compared with Spherical CNNs and the ground truth. }
    \label{Power Spectrum}
\end{figure}

\paragraph{Result}
Figure~\ref{Power Spectrum} compares the power spectrum of two models' predicted B-unlensed map with ground truth. We can see Spherical CNNs always underestimates the ground truth when degree $l$ is larger than $20$. NES can capture more high-frequency information of data in each block during training. The estimated power spectrum of the model is consistent with the ground truth even at a large degree of $l\geq 100$.

\section{Conclusion}
We develop a Needlet approximate Equivariance Spherical CNN using multiscale representation systems on the sphere and rotation group. The needlet convolution inherits the multiresolution analysis ability from needlet transforms and allows rotation invariance in network propagation. Wavelet shrinkage is used as a network activation to filter out the high-pass redundancy, which helps improve the robustness of the network. The shrinkage brings controllable equivariance error for the needlet CNN, which is small when the scale is high. Empirical study shows the proposed model can achieve excellent performance on real scientific problems.

\section*{Acknowledgements}
YW acknowledges support from the Shanghai Municipal Science and Technology Major Project (2021SHZDZX0102) and Huawei-SJTU ExploreX Funding (SD6040004/034). We are also grateful to the anonymous reviewers for their feedback.

\bibliographystyle{icml2022}
\bibliography{main}

\appendix
\section{Generalized Fourier Transform}\label{GFT}
Denoting the manifold ($\mathbb{S}^2$ or $\mathrm{SO(3)}$) by $\mathcal{M}$. The basis functions are spherical harmonics ($Y_m^\ell(R)$) and Wigner D-functions ($D_{mn}^\ell(R)$) for $\mathbb{S}^2$ and $\mathrm{SO(3)}$ respectively. Denote the basis functions as $u_\ell$. We can write the generalized Fourier transform of a function $f: \mathcal{M}\rightarrow \mathbb{R}$ with quadrature rule sampling at scale $j$ as 
\begin{align*}
  \widehat{f}_{\ell} =\left \langle f,u_\ell \right \rangle &= \int_\mathcal{M}f(\bx)\overline{u_\ell(\bx)}\dd\bx\\
  &= \sum_{k=0}^{N_j}f(\boldsymbol{x}_{j,k})\sqrt{\omega_{j,k}}\overline{u_\ell(\boldsymbol{x}_{j,k})}.
\end{align*}
The inverse Fourier transforms on $\mathbb{S}^2$ and $\mathrm{SO(3)}$ are as follows,
\begin{align*}
    f(R) &= [\mathcal{F}^{-1}\hat{f}](R)=\sum_{\ell=0}^\infty(2\ell+1)\sum_{m=-\ell}^{\ell}\hat{f}^\ell_{m}Y_m^\ell(R)\\
    f(R) &= [\mathcal{F}^{-1}\hat{f}](R)=\sum_{\ell=0}^\infty(2\ell+1)\sum_{m=-\ell}^{\ell}\sum_{n=-\ell}^{\ell}\hat{f}^\ell_{mn}D_{mn}^\ell(R).
\end{align*}
Let $(\alpha,\beta )$ with $\alpha\in [0,2\pi]$ and $\beta\in[0,\pi]$ be the spherical polar coordinates for $x\in \mathbb{S}^2$. The spherical harmonics can be explicitly written as 
$$Y_{\ell, m}(\alpha, \beta):=\sqrt{\frac{2 \ell+1}{4 \pi} \frac{(\ell-m) !}{(\ell+m) !}} P_{\ell}^{(m)}(\cos \beta) e^{i m \alpha},$$
where $P_{\ell}^{(m)}(t)$ is the associated Legendre Polynomial of degree $\ell$ and order $m$, $m=-\ell, \ldots, \ell, \ell=0,1, \ldots$. We use the ZYZ Euler parameterization for $\mathrm{SO(3)}$. An element $R\in \mathrm{SO(3)}$ can be parameterized by $R(\alpha, \beta,\gamma)$ with $\alpha\in [0,2\pi]$, $\beta\in [0,\pi]$ and $\gamma\in [0,2\pi]$.
There exists a general relationship between Wigner D-functions and spherical harmonics:
$$D_{m s}^{\ell}(\alpha, \beta,-\gamma)=(-1)^{s} \sqrt{\frac{4 \pi}{2 \ell+1}}_{s} Y_{\ell}^{m}(\beta, \alpha) e^{i s \gamma}.$$

\section{Needlets on $\mathbb{S}^2$ and $\mathrm{SO(3)}$}\label{needlet}
Tight framelets on manifold $\mathcal{M}$ is defined by a filter bank (a set of complex-valued filters) $\eta:=\left\{a ; b^{1}, \ldots, b^{r}\right\} \subset l_{1}(\mathbb{Z}):=\left\{h=\left\{h_{k}\right\}_{k \in \mathbb{Z}} \subset \mathbb{C}: \sum_{k \in \mathbb{Z}}\left|h_{k}\right|<\infty\right\}$ and a set of associated scaling functions, $\Psi=\{\alpha;\beta^1,\cdots,\beta^r\} \subset L_1(\mathbb{R})$, which is a set of complex-valued functions on the real axis satisfying the following equations, for $r=1,\ldots, n, \xi \in \mathbb{R}$,
$$\widehat{\alpha}(2\xi)=\widehat{a}(\xi) \widehat{\alpha}(\xi),\quad
\widehat{\beta^{n}}(2 \xi)=\widehat{b^{n}}(\xi) \widehat{\alpha}(\xi).$$
Here $a(\cdot)$ is called the $\emph{low-pass}$ filter and $b^{(r)}(\cdot)$ are $\emph{high-pass}$ filters. Let $\left\{\left(\boldsymbol{u}_{\ell=1}^{\infty}, \lambda_{\ell} \right)\right\}_{\ell}$ be the eigenvalue and eigenvector pairs for $L_2(\mathcal{M})$. The framelets at scale level $j=1,\cdots,J$ for manifold $\mathcal{M}$ are generated with the above scaling functions and orthonormal eigen-pairs as 
\begin{equation*}
\begin{aligned}
\boldsymbol{\varphi}_{j,\boldsymbol{y}}(\boldsymbol{x}) &=\sum_{\ell=1}^{\Lambda_j} \widehat{\alpha}\left(\frac{\lambda_{\ell}}{2^{j}}\right) \overline{\boldsymbol{u}_{\ell}(\boldsymbol{y})} \boldsymbol{u}_{\ell}(\boldsymbol{x}) \\
\boldsymbol{\psi}_{j, \boldsymbol{y}}^{n}(\boldsymbol{x}) &=\sum_{\ell=1}^{\Lambda_j} \widehat{\beta^{(n)}}\left(\frac{\lambda_{\ell}}{2^{j}}\right) \overline{\boldsymbol{u}_{\ell}(\boldsymbol{y})} \boldsymbol{u}_{\ell}(\boldsymbol{x})
\end{aligned}
\end{equation*}
$\boldsymbol{\varphi}_{j,\boldsymbol{y}}(\boldsymbol{x})$ and $\{\boldsymbol{\psi}_{j, \boldsymbol{y}}^{n}(\boldsymbol{x})\}_{n=1}^r$ are low-pass and high -pass framelets at scale $j$ at point $\boldsymbol{y}\in\mathcal{M}$. $\Lambda_j$ is the bandwidth of scale level $j$ and $n=1,\cdots,r$ depending on the support of scaling functions $\widehat{\alpha}$ and $\widehat{\beta^{(n)}}$. 

Needlets are a type of framelets on the sphere $\mathrm{\mathbb{S}}^d$ associated with a quadrature rule and a specific filter bank. This type of framelets can be generalized to rotation group $\mathrm{SO(3)}$ with the same filter bank. For simplicity, we consider the filter bank $\boldsymbol{\eta}=\{a;b^1,b^2\}$ with two high-pass filters. We define the filter bank $\boldsymbol{\eta}=\{a;b^1,b^2\}$ by their Fourier series as follows.

\begin{equation}\label{filter_bank}
\begin{aligned}
\widehat{a}(\xi)&:=\left\{\begin{array}{ll}
1, & |\xi|<\frac{1}{8}, \\[1mm]
\cos \left(\frac{\pi}{2} \nu(8|\xi|-1)\right), & \frac{1}{8} \leqslant |\xi| \leqslant \frac{1}{4}, \\[1mm]
0, & \frac{1}{4}<|\xi| \leqslant \frac{1}{2},
\end{array}\right. \\[2mm]
\widehat{b^{1}}(\xi)&:=\left\{\begin{array}{ll}
0, & |\xi|<\frac{1}{8}, \\[1mm]
\sin \left(\frac{\pi}{2} \nu(8|\xi|-1)\right), & \frac{1}{8} \leqslant |\xi| \leqslant \frac{1}{4}, \\[1mm]
\cos \left(\frac{\pi}{2} \nu(4|\xi|-1)\right), & \frac{1}{4}<|\xi| \leqslant \frac{1}{2},
\end{array}\right. \\[2mm]
\widehat{b^{2}}(\xi)&:=\left\{\begin{array}{ll}
0, & |\xi|<\frac{1}{4} \\[1mm]
\sin \left(\frac{\pi}{2} \nu(4|\xi|-1)\right), & \frac{1}{4} \leqslant |\xi| \leqslant \frac{1}{2},
\end{array}\right.
\end{aligned}
\end{equation}

where $$
\nu(t):=\chi_{3}(t)^{2}=t^{4}\left(35-84 t+70 t^{2}-20 t^{3}\right), \quad t \in \mathbb{R}.
$$
It can be verified that 
$$|\widehat{a}(\xi)|^{2}+\left|\widehat{b^{1}}(\xi)\right|^{2}+\left|\widehat{b^{2}}(\xi)\right|^{2}=1 \quad \forall \xi \in[0,1 / 2].$$
Then, the associated needlets generators $\Psi=\{\alpha;\beta^1,\beta^2\}$ is explicitly given by 

\begin{equation}\label{needlet_gen}
\begin{aligned}
\widehat{\alpha}(\xi)&=\left\{\begin{array}{ll}
1, & |\xi|<\frac{1}{4}, \\
\cos \left(\frac{\pi}{2} \nu(4|\xi|-1)\right), & \frac{1}{4} \leqslant|\xi| \leqslant \frac{1}{2}, \\
0, & \text { else },
\end{array}\right. \\[2mm]
\widehat{\beta^{1}}(\xi)&=\left\{\begin{array}{ll}
\sin \left(\frac{\pi}{2} \nu(4|\xi|-1)\right), & \frac{1}{4} \leqslant|\xi|<\frac{1}{2}, \\[1mm]
\cos ^{2}\left(\frac{\pi}{2} \nu(2|\xi|-1)\right), & \frac{1}{2} \leqslant|\xi| \leqslant 1, \\[1mm]
0, & \text { else },
\end{array}\right. \\[2mm]
\widehat{\beta^{2}}(\xi)&=\left\{\begin{array}{ll}
0, & |\xi|<\frac{1}{2} \\[1mm]
\frac{1}{2}\sin \left(\pi\nu(2|\xi|-2)\right), & \frac{1}{2} \leqslant|\xi| \leqslant 1 \\[1mm]
0, & \text { else.}
\end{array}\right.
\end{aligned}
\end{equation}

The framelet coefficients $v_{j,k}$ represent low-pass coefficients, and $w^n_{j,k}$ represent high-pass coefficients. They are defined as $\langle\boldsymbol{\varphi}_{j,k},f \rangle$ and $\langle\boldsymbol{\psi}_{j,k}^{n},f\rangle$ respectively. We can calculate the coefficients in the Fourier space, as shown in Eq.~(\ref{freq_vw_define}) below.
 \begin{equation}
\begin{aligned}
    v_{j, k}&=\sum_{\ell=0}^{\Lambda_{j}} \widehat{f_{\ell}} \overline{\widehat{\alpha}\left(\frac{\lambda_{\ell}}{2^{j}}\right)} \sqrt{\omega_{j, k}} u_{\ell}\left(\boldsymbol{x}_{j, k}\right)\\
    w_{j-1, k}^{n}&=\sum_{\ell=0}^{\Lambda_{j}} \widehat{f}_{\ell} \overline{\widehat{\beta^{n}}\left(\frac{\lambda_{\ell}}{2^{j-1}}\right)} \sqrt{\omega_{j, k}} u_{\ell}\left(\boldsymbol{x}_{j, k}\right)
\end{aligned}
\label{v&w}
\end{equation}
\begin{equation}
\widehat{v}_{j, \ell}=\widehat{f}_{\ell} \overline{\widehat{\alpha}\left(\frac{\lambda_{\ell}}{2^{j}}\right)}
,\quad \widehat{w}_{j-1, \ell}^{n}=\widehat{f}_{\ell} \overline{\widehat{\beta^{n}}\left(\frac{\lambda_{\ell}}{2^{j-1}}\right)}.
\label{freq_vw_define}
\end{equation}

\section{Polynomial-exact Quadrature Rule}
Let $Q_{N_{j}}=\left\{\left(\omega_{j, k}, \boldsymbol{x}_{j, k}\right) \in \mathbb{R} \times \mathcal{M}: k=0, \ldots, N_{j}\right\}$ be a polynomial-exact quadrature rule at scale $j$ with $N_j$ weights $\omega_{j,k}\in \mathbb{R}$ and $N_j$ points $\boldsymbol{x}_{j,k}\in\mathcal{M}$. We use Gauss-Legendre tensor product rule which is generated by the tensor product of the Gauss-Legendre nodes on the interval $[-1,1]$ and equi-spaced nodes in the other dimension with non-equal weights. By using polynomial-exact quadrature rule, the integral of polynomial of degree less than a certain $l$ yields zero numerical error. We use $Q_{N_j}$ and $Q_{N_{j+1}}$ to discretize the continuous framelets $\boldsymbol{\varphi}_{j,\boldsymbol{y}}$ and $\boldsymbol{\psi}_{j, \boldsymbol{y}}^{n} (n=1, \cdots, r)$  in Eq.~(\ref{framelet}) respectively as follows.
\begin{equation}
\begin{aligned}
\boldsymbol{\varphi}_{j, k}(\boldsymbol{x})&=\sqrt{\omega_{j, k}} \sum_{\ell=1}^{\Lambda_j} \widehat{\alpha}\left(\frac{\lambda_{\ell}}{2^{j}}\right) \overline{u_{\ell}\left(\boldsymbol{x}_{j, k}\right)} u_{\ell}(\boldsymbol{x}) \\
\boldsymbol{\psi}_{j, k}^{n}(\boldsymbol{x})&=\sqrt{\omega_{j+1, k}} \sum_{\ell=1}^{\Lambda_j} \widehat{\beta^{n}}\left(\frac{\lambda_{\ell}}{2^{j}}\right) \overline{u_{\ell}\left(\boldsymbol{x}_{j+1, k}\right)} u_{\ell}(\boldsymbol{x}).
\end{aligned}
\label{framelet}
\vspace{-0.2cm}
\end{equation}
The generalized Fourier coefficients $\widehat{f}_\ell$, $\ell=0,1,\cdots ,\Lambda_j$, are calculated as 
$$
  \widehat{f}_{\ell} 
  = \int_\mathcal{M}f(\bx)\overline{u_\ell(\bx)}\dd\bx= \sum_{k=0}^{N_j}f(\boldsymbol{x}_{j,k})\sqrt{\omega_{j,k}}\overline{u_\ell(\boldsymbol{x}_{j,k})}.
$$
For (\ref{framelet}) with scaling functions (\ref{needlet_gen}) when the support of the scaling functions is in $[-1,1]$, the quadrature rule is needed to set exact for degree $\leq 2^{j+1}$ at scale $j$.

The inverse generalized Fourier transform is defined by 
\begin{equation*}
    f(\boldsymbol{x}_{j,k}) = \sum_{\ell=1}^{\Lambda_j}\widehat{f}_{\ell}\sqrt{\omega_{j,k}}u_\ell(\boldsymbol{x}_{j,k}),\; k=0,1,\cdots,N_j
\end{equation*}
For $\mathbb{S}^2$ signal, basis functions $\{u_\ell\}_\ell$ are usually the spherical harmonics $Y_m^\ell: \mathbb{S}^2\rightarrow\mathbb{C}$ indexed by $l\geq 0$ and $-l\leq m\leq l$ . For $\mathrm{SO(3)}$ signal, basis functions $\{u_\ell\}_\ell$ are usually Wigner-D functions $D_{mn}^\ell:\mathrm{SO(3)}\rightarrow \mathbb{C}$ indexed by $l\geq 0$ and $-l\leq m,n \leq l$. More details about $Y_m^l$ and $D_{mn}^{\ell}$ are given in Appendix \ref{GFT}.

We let $v_{j,k}$ represent low-pass coefficients, and $w^n_{j,k}$ represent high-pass coefficients, which are defined as $\langle\boldsymbol{\varphi}_{j,k},f \rangle$ and $\langle\boldsymbol{\psi}_{j,k}^{n},f\rangle$ respectively. We can calculate the coefficients in the Fourier space:
\begin{equation}
\widehat{v}_{j, \ell}=\widehat{f}_{\ell} \overline{\widehat{\alpha}\left(\frac{\lambda_{\ell}}{2^{j}}\right)}
,\quad \widehat{w}_{j-1, \ell}^{n}=\widehat{f}_{\ell} \overline{\widehat{\beta^{n}}\left(\frac{\lambda_{\ell}}{2^{j-1}}\right)}.
\label{freq_vw_define_2}
\vspace{-0.2cm}
\end{equation}

\section{Conditions of Tightness of Needlet System}\label{tightness}
In the implementation, we need the discrete version of needlets. We get discrete needlets $\boldsymbol{\varphi}_{j,k}(\boldsymbol{x})$ and $\boldsymbol{\psi}_{j,k}^n(\boldsymbol{x})$ at scale $j$ with quadrature rule sampling \cite{wang2017fully}. Let $\mathcal{Q}=\{Q_{N_j}\}_{j\geqslant J}$, the set of quadrature rules at scale $j\geqslant J$, and define the needlet system as 
$$
\mathrm{FS}_{J}(\Psi, \mathcal{Q}):=\left\{\varphi_{J, k}\right\} \cup\left\{\psi_{j, k}^n:j \geqslant J;n \leqslant r\right\},
$$
where $k=0, \ldots, N_{J}$ for $\varphi_{J,k}$ and $k=0, \ldots, N_{j+1}$ for $\varphi_{J,k}^n$. The tight needlet system can be constructed on a general Riemannian manifold.
\begin{definition}
Let $\mathcal{M}$ be the compact and smooth Riemannian manifold.
The needlet system $\mathrm{FS}_{J}(\Psi, \mathcal{Q})$ is said to be \textbf{tight} for $f\in L_2(\mathcal{M})$ if and only if  $\mathrm{FS}_{J}(\Psi, \mathcal{Q})\subset L_2(\mathcal{M})$ and
\begin{equation}
    f=\sum_{k=0}^{N_{J}}\left\langle f, \boldsymbol{\varphi}_{J, k}\right\rangle \boldsymbol{\varphi}_{J, k}+\sum_{j=J}^{\infty} \sum_{k=0}^{N_{j+1}} \sum_{n=1}^{r}\left\langle f, \boldsymbol{\psi}_{j, k}^{n}\right\rangle \boldsymbol{\psi}_{j, k}^{n}
    \label{definition}
\end{equation}
or equivalently,
$$
\|f\|_{L_{2}(\mathcal{M})}^{2}=\sum_{k=0}^{N_{J}}\left|\left\langle f, \boldsymbol{\varphi}_{J, k}\right\rangle\right|^{2} +\sum_{j=J}^{\infty} \sum_{k=0}^{N_{j+1}} \sum_{n=1}^{r}\left|\left\langle f, \boldsymbol{\psi}_{j, k}^{n}\right\rangle\right|^{2}.
$$
\end{definition}
The tightness of needlet system on $\mathbb{S}^2$ is given in \citep{wang2020tight}. Here, the following theorem gives the equivalence conditions of needlet systems $\mathrm{FS}_{J}(\Psi, \mathcal{Q})$ on $\mathrm{SO(3)}$ to be tight frames for $L_2(\mathrm{SO(3)})$.
\begin{theorem}
Let $J_0\in \mathbb{Z}$ be an integer and $\boldsymbol{\Psi}:=\{\alpha;\beta^1,\ldots,\beta^r\}\subset L_1(\mathbb{R})$ with $r\geqslant 1$ be a set the needlet generators associated with the filter bank $\boldsymbol{\eta}:=\{a;b_1,\ldots,b_r\}\subset l_1(\mathbb{Z})$ as Eq.~(\ref{filter_bank}) and Eq.~(\ref{needlet_gen}).  Let $\mathcal{Q}=\{Q_{N_j}\}_{j\geqslant J}$, the set of quadrature rules $Q_{N_j} = \{(\omega_{j,k},\boldsymbol{x}_{j,k})\in \mathbb{R}\times \mathrm{SO(3)}:k=0,\ldots, N_j\}$. $\mathrm{FS}_{J}(\Psi, \mathcal{Q})\subset L_2(\mathrm{SO(3)})$ is the needlet system. Then, the following statements are equivalent.
\begin{enumerate}
    \item The needlet system $\mathrm{FS}_{J}(\Psi, \mathcal{Q})$ is a tight frame for $L_2(\mathrm{SO(3)})$ for all $J\geqslant J_0$, i.e., Eq.~(\ref{definition}) holds for all $J\geqslant J_0$.
    \item For all $f\in L_2(\mathrm{SO(3)})$, the following identities hold for all $j \geqslant J_{0}$:
    \begin{equation}
        \begin{aligned}
        &\lim _{j \rightarrow\infty}\left\|\sum_{k=0}^{N_{j}}\left\langle f, \boldsymbol{\varphi}_{j, k}\right\rangle \boldsymbol{\varphi}_{j, k}-f\right\|_{L_{2}(\mathrm{SO(3)})}=0\\
        &\sum_{k=0}^{N_{j+1}}\left\langle f, \boldsymbol{\varphi}_{j+1, k}\right\rangle \boldsymbol{\varphi}_{j+1, k}=\\
        &\qquad\sum_{k=0}^{N_{j}}\left\langle f, \boldsymbol{\varphi}_{j, k}\right\rangle \boldsymbol{\varphi}_{j, k}+
        \sum_{k=0}^{N_{j+1}} \sum_{n=1}^{r}\left\langle f, \boldsymbol{\psi}_{j, k}^{n}\right\rangle \boldsymbol{\psi}_{j, k}^{n}.
        \end{aligned}
        \label{C.1:condition_2}
    \end{equation}
    \item For all $f\in L_2(\mathrm{SO(3)})$, the following identities hold for all $j \geqslant J_{0}$:
    \begin{equation}
        \begin{aligned}
    &\lim _{j \rightarrow \infty} \sum_{k=0}^{N_{j}}\left|\left\langle f, \boldsymbol{\varphi}_{j, k}\right\rangle\right|^{2}=\|f\|_{L_{2}(\mathrm{SO(3)})}^{2}, \\
    &\sum_{k=0}^{N_{j+1}}\left|\left\langle f, \boldsymbol{\varphi}_{j+1, k}\right\rangle\right|^{2}\\
    &\qquad =\sum_{k=0}^{N_{j}}\left|\left\langle f, \boldsymbol{\varphi}_{j, k}\right\rangle\right|^{2}
    +\sum_{k=0}^{N_{j+1}} \sum_{n=1}^{r}\left|\left\langle f, \psi_{j, k}^{n}\right\rangle\right|^{2}.
        \end{aligned}
    \label{C.1:condition_3}
    \end{equation}
    \item The generators in $\boldsymbol{\Psi}$ satisfy, for all $\ell, \ell^\prime\geqslant 0$ and $j\geqslant J_0$,
    \begin{equation}
        \begin{aligned}
    &\lim _{j \rightarrow \infty} \overline{\widehat{\alpha}\left(\frac{\lambda_{\ell}}{2^{j}}\right)} \widehat{\alpha}\left(\frac{\lambda_{\ell^{\prime}}}{2^{j}}\right) \mathcal{U}_{\ell, \ell^{\prime}}\left(Q_{N_{j}}\right)= \\
    &\delta_{\ell, \ell^{\prime}}\left[\overline{\widehat{\alpha}\left(\frac{\lambda_{\ell}}{2^{j+1}}\right)} \widehat{\alpha}\left(\frac{\lambda_{\ell^{\prime}}}{2^{j+1}}\right)
    -\sum_{n=1}^{r} \overline{\widehat{\beta^{n}}\left(\frac{\lambda_{\ell}}{2^{j}}\right)} \widehat{\beta^{n}}\left(\frac{\lambda_{\ell^{\prime}}}{2^{j}}\right)\right]\times\\
    &\quad \mathcal{U}_{\ell, \ell^{\prime}}\left(Q_{N_{j}}\right)
    =\overline{\widehat{\alpha}\left(\frac{\lambda_{\ell}}{2^{j}}\right)} \widehat{\alpha}\left(\frac{\lambda_{\ell^{\prime}}}{2^{j}}\right)  \mathcal{U}_{\ell, \ell^{\prime}}\left(Q_{N_{j+1}}\right),
    \end{aligned}
    \label{C.1:condition_4}
    \end{equation} 
    where 
    $$
\mathcal{U}_{\ell, \ell^{\prime}}\left(Q_{N_{j}}\right):=\sum_{k=0}^{N_{j}} \omega_{j, k} u_{\ell}\left(\boldsymbol{x}_{j, k}\right) \overline{u_{\ell^{\prime}}\left(\boldsymbol{x}_{j, k}\right)}.
$$
    \item The filters in the filter bank $\boldsymbol{\eta}$ satisfy, for all $j\geqslant J_0+1$,
    \begin{equation}
        \begin{aligned}
&\mathcal{U}_{\ell, \ell^{\prime}}\left(Q_{N_{j}}\right) =
\overline{\widehat{a}\left(\frac{\lambda_\ell}{2^{j}}\right)} \widehat{a}\left(\frac{\lambda_ \ell^{\prime}}{2^{j}}\right) \mathcal{U}_{\ell, \ell^{\prime}}\left(Q_{N_{j-1}}\right)\\
&\qquad +\sum_{n=1}^{r} \overline{b^{n}\left(\frac{\lambda_{\ell}}{2^{j}}\right)} \hat{b^{n}}\left(\frac{\lambda_{\ell^{\prime}}}{2^{j}}\right) \mathcal{U}_{\ell, \ell^{\prime}}\left(Q_{N_{j}}\right),         
        \end{aligned}
        \label{C.1:condition_5}
    \end{equation}
where 
$$
    \sigma_{\alpha, \bar{\alpha}}^{j}:=\left\{\left(\ell, \ell^{\prime}\right) \in \mathbb{N}_{0} \times \mathbb{N}_{0}: \widehat{\widehat{\alpha}\left(\frac{\lambda_{\ell}}{2^{j}}\right)} \widehat{\alpha}\left(\frac{\lambda_{\ell^{\prime}}}{2^{j}}\right) \neq 0\right\}.
$$
\end{enumerate}
\end{theorem}
\begin{proof}
1$\iff$2: We define projections $\boldsymbol{P}_{\Phi_{j}}$ and $\boldsymbol{P}_{\Psi_j}$ as
\begin{equation*}
    \begin{aligned}
    \boldsymbol{P}_{\Phi_{j}}(f)&:=\sum_{k=0}^{N_j}\left\langle f, \varphi_{j, k}\right\rangle \boldsymbol{\varphi}_{j, k}\\
    \boldsymbol{P}_{\Psi_{j}^{n}}(f)&:=\sum_{k=0}^{N_j}\left\langle f, \psi_{j, k}^{n}\right\rangle \psi_{j, k}^{n}.
    \end{aligned}
\end{equation*}
Since the needlet system $\mathrm{FS}_J$ is tight for $L_2(\mathrm{SO(3})$ for all $J\geqslant J_0$,
\begin{equation*}
\begin{aligned}
f &= \boldsymbol{P}_{\Phi_j}(f) + \sum_{j=J}^{\infty}\sum_{n=1}^{r}\boldsymbol{P}_{\Psi_j^n}(f) \\
&= \boldsymbol{P}_{\Phi_{J+1}}(f)+\sum_{j=J+1}^\infty\sum_{n=1}^r\boldsymbol{P}_{\Psi_j^n}(f)
\end{aligned}
\end{equation*}
for all $f\in L_2(\mathrm{SO(3)})$ and all $J\geqslant J_0$. Thus, in $L_2$ sense,
$$\boldsymbol{P}_{\Phi_{J+1}}(f)=\boldsymbol{P}_{\Phi_J}(f) + \sum_{n=1}^r\boldsymbol{P}_{\Psi_j^n}(f).$$
Recursively, we have
\begin{equation}
    \boldsymbol{P}_{\Phi_{m+1}}(f) = \boldsymbol{P}_{\Phi_J}(f)+\sum_{j=J}^m\sum_{n=1}^r\boldsymbol{P}_{\Psi_j^n}(f)
    \label{mm_eq}
\end{equation}
for all $m\geqslant J$ and $J\geqslant J_0$. Let $m\rightarrow \infty$, in $L_2$ sense, we obtain
\begin{equation}
    \lim_{m\rightarrow\infty}\boldsymbol{P}_{\Phi_{m+1}}(f) = \boldsymbol{P}_{\Phi_J}(f) + \sum_{j=J}^\infty\sum_{n=1}^r\boldsymbol{P}_{\Psi_J^n}(f) = f,
    \label{mm_eq_inf}
\end{equation}
which is Eq.~(\ref{C.1:condition_2}). Consequently, we proof that 1$\iff$2. Conversely, by Eq.~(\ref{C.1:condition_2}) and Eq.~(\ref{mm_eq}), let $m\rightarrow\infty$, we can get Eq.~(\ref{definition}). Thus, 2$\iff$1.

2$\iff$3: We can deduce the equivalence between 2 and 3 by  polarization identity.

3$\iff$4: For $f\in L_2(\mathrm{SO(3)})$, by the formulas in Eq.~(\ref{framelet}) and the orthonormality of Wigner D-functions, which are the basis of $\mathrm{SO(3)}$, we obtain 
\begin{equation}
\begin{aligned}
 \left\langle f, \boldsymbol{\varphi}_{j, k}\right\rangle&=\sqrt{\omega_{j, k}} \sum_{\ell=0}^{\infty} \widehat{\widehat{\alpha}\left(\frac{\lambda_{\ell}}{2^{j}}\right)} \widehat{f}_{\ell} u_{\ell}\left(\boldsymbol{x}_{j, k}\right)\\
 \left\langle f, \boldsymbol{\psi}_{j, k}^{n}\right\rangle&=\sqrt{\omega_{j+1, k}} \sum_{\ell=0}^{\infty} \widehat{\widehat{\beta^{n}}\left(\frac{\lambda_{\ell}}{2^{j}}\right)} \widehat{f}_{\ell} u_{\ell}\left(\boldsymbol{x}_{j+1, k}\right).
\end{aligned}
\end{equation}
Then, we have 
\begin{equation*}
    \begin{aligned}
&\sum_{k=0}^{N_{j}}\left|\left\langle f, \varphi_{j, k}\right\rangle\right|^{2} =\sum_{k=0}^{N_{j}} \omega_{j, k}\left|\sum_{\ell=0}^{\infty} \overline{\widehat{\alpha}\left(\frac{\lambda_{\ell}}{2^{j}}\right)} \widehat{f}_{\ell} u_{\ell}\left(\boldsymbol{x}_{j, k}\right)\right|^{2} \\
&=\sum_{\ell=0}^{\infty} \sum_{\ell^{\prime}=0}^{\infty} \widehat{f_{\ell}} \overline{\hat{f}_{\ell^{\prime}}} \overline{\widehat{\alpha}\left(\frac{\lambda_{\ell}}{2^{j}}\right)} \widehat{\alpha}\left(\frac{\lambda_{\ell^{\prime}}}{2^{j}}\right) \sum_{k=0}^{N_{j}} \omega_{j, k} u_{\ell}\left(\boldsymbol{x}_{j, k}\right) \overline{u_{\ell^{\prime}}\left(\boldsymbol{x}_{j, k}\right)} \\
&=\sum_{\ell=0}^{\infty} \sum_{\ell^{\prime}=0}^{\infty} \widehat{f}_{\ell} \overline{\widehat{f}_{\ell^{\prime}}} \overline{\widehat{\alpha}\left(\frac{\lambda_{\ell}}{2^{j}}\right)} \widehat{\alpha}\left(\frac{\lambda_{\ell^{\prime}}}{2^{j}}\right) \mathcal{U}_{\ell, \ell^{\prime}}\left(Q_{N_{j}}\right) \\
&=\sum_{\ell=0}^{\infty}\left|\widehat{f}_{\ell}\right|^{2}\left|\widehat{\alpha}\left(\frac{\lambda_{\ell}}{2^{j}}\right)\right|^{2} \mathcal{U}_{\ell, \ell}\left(Q_{N_{j}}\right)\\
&+\sum_{\ell=0}^{\infty} \sum_{\ell^{\prime}=0, \ell^{\prime} \neq \ell}^{\infty} \widehat{f}_{\ell} \overline{\widehat{f}_{\ell^{\prime}}}\overline{\widehat{\alpha}\left(\frac{\lambda_{\ell}}{2^{j}}\right)} \widehat{\alpha}\left(\frac{\lambda_{\ell^{\prime}}}{2^{j}}\right) \mathcal{U}_{\ell, \ell^{\prime}}\left(Q_{N_{j}}\right).
\end{aligned}
\end{equation*}
Then, Eq.~(\ref{C.1:condition_3}) holds only and if only Eq.~(\ref{C.1:condition_4}) holds.

4$\iff$5: By condition \ref{relations}, we have 
\begin{equation*}
\begin{aligned}
&\overline{\widehat{\alpha}\left(\frac{\lambda_{\ell}}{2^{j-1}}\right)} \widehat{\alpha}\left(\frac{\lambda_{\ell^{\prime}}}{2^{j-1}}\right) \mathcal{U}_{\ell, \ell^{\prime}}\left(Q_{N_{j-1}}\right)+\\
&\sum_{n=1}^{r} \overline{\widehat{\beta^{n}}\left(\frac{\lambda_{\ell}}{2^{j-1}}\right)} \widehat{\beta^{n}}\left(\frac{\lambda_{\ell^{\prime}}}{2^{j-1}}\right) \mathcal{U}_{\ell, \ell^{\prime}}\left(Q_{N_{j}}\right)\\
&=\overline{\widehat{a}\left(\frac{\lambda_{\ell}}{2^{j}}\right)} \widehat{a}\left(\frac{\lambda_{\ell^{\prime}}}{2^{j}}\right) \mathcal{U}_{\ell, \ell^{\prime}}\left(Q_{N_{j-1}}\right)\widehat{\alpha}\left(\frac{\lambda_{\ell}}{2^{j}}\right) \widehat{\alpha}\left(\frac{\lambda_{\ell^{\prime}}}{2^{j}}\right)\\
&+\sum_{n=1}^{r} \overline{\widehat{b^{n}}\left(\frac{\lambda_{\ell}}{2^{j}}\right)} \widehat{b^{n}}\left(\frac{\lambda_{\ell^{\prime}}}{2^{j}}\right) \mathcal{U}_{\ell, \ell^{\prime}}\left(Q_{N_{j}}\right)\widehat{\alpha}\left(\frac{\lambda_{\ell}}{2^{j}}\right) \widehat{\alpha}\left(\frac{\lambda_{\ell^{\prime}}}{2^{j}}\right)
\end{aligned}
\end{equation*}
Therefore, the equivalence between Eq.~(\ref{C.1:condition_4}) and Eq.~(\ref{C.1:condition_5}) holds.
\end{proof}

\section{Needlet Decompostion and Reconstruction}\label{decom}
\begin{algorithm}
\caption{Decomposition of Multi-Level Needlet Transform}
\label{algo:decomp.multi.level}
\begin{algorithmic}
\STATE{\bfseries Input: }$v_J$ -- a $(\Lambda_J,N_J)$-sequence
\STATE{\bfseries Output: }($\{w_{J-1}^n,w_{J-2}^n,\ldots,w_{J_0}^n\}_{n=1}^r,v_{J_0}$)
\STATE $v_J \rightarrow \hat{v}_J$
\FOR{$j\leftarrow J$ \textbf{to} $J_0+1$}
\STATE $\widehat{v}_{j-1} \longleftarrow \widehat{v}_{j, \cdot}\overline{\widehat{a}}\left(2^{-j} \lambda .\right)$
\FOR{$n\leftarrow 1$ \textbf{to} $r$}
\STATE $\widehat{w}_{j-1}^{n} \longleftarrow \widehat{v}_{j, \cdot}\overline{\widehat{b^{n}}}\left(2^{-j} \lambda_{.}\right)$
\STATE $w_{j-1}^n \leftarrow \hat{w}_{j-1}^n$
\ENDFOR
\ENDFOR
\STATE $v_{J_0}\leftarrow \hat{v}_{J_0}$
\end{algorithmic}
\end{algorithm}

\begin{algorithm}
\caption{Reconstruction of Multi-Level Needlet Transform}
\label{algo:rec.multi.level}
\begin{algorithmic}
\STATE{\bfseries Input: }($\{w_{J-1}^n,w_{J-2}^n,\ldots,w_{J_0}^n\}_{n=1}^r,v_{J_0}$)
\STATE{\bfseries Output: }$v_J$ -- a $(\Lambda_J,N_J)$-sequence
\STATE $\hat{v}_{J_0} \leftarrow {v}_{J_0}$
\FOR{$j\leftarrow J_0+1$ \textbf{to} $J$}
\FOR{$n\leftarrow 1$ \textbf{to} $r$}
\STATE $\widehat{w}_{j-1}^{n} \longleftarrow w_{j-1}^n$
\ENDFOR
\STATE $\widehat{\mathrm{v}}_{j} \longleftarrow\left(\widehat{\mathrm{v}}_{j-1, \cdot}\right) \widehat{a}\left(2^{-j} \lambda .\right)+\sum_{n=1}^{r} \widehat{\mathrm{w}}_{j, \cdot}^{n} \widehat{b^{n}}\left(2^{-j} \lambda.\right)$
\ENDFOR
\STATE $v_{J}\leftarrow\hat{v}_{J}$
\end{algorithmic}
\end{algorithm}
As Eq.~(\ref{v&w}), we have $v_{j,k}$ and $w_{j-1,k}^n$ as $(\Lambda_j,N_j)$ sequences. We have the following \emph{decomposition relation}:
\begin{equation*}
    \begin{aligned}
v_{j-1, k} &=\sum_{\ell=0}^{\Lambda_{j-1}} \widehat{f_{\ell}} \overline{\widehat{\alpha}\left(\frac{\lambda_{\ell}}{2^{j-1}}\right)} \sqrt{\omega_{j-1, k}} u_{\ell}\left(\boldsymbol{x}_{j-1, k}\right) \\
&=\sum_{\ell=0}^{\Lambda_{j-1}} \widehat{f_{\ell}} \overline{\widehat{\alpha}\left(\frac{\lambda_{\ell}}{2^{j}}\right)} \overline{\widehat{a}\left(\frac{\lambda_{\ell}}{2^{j}}\right)} \sqrt{\omega_{j-1, k}} u_{\ell}\left(\boldsymbol{x}_{j-1, k}\right) \\
&=\sum_{\ell=0}^{\Lambda_{j}}\widehat{v_{j,\ell}} \overline{\widehat{\alpha}\left(\frac{\lambda_{\ell}}{2^{j}}\right)} \sqrt{\omega_{j-1, k}} u_{\ell}\left(\boldsymbol{x}_{j-1, k}\right) \\
&=\left[\left(v_{j} *_{j} a^{\star}\right) \downarrow_{j}\right](k),
\end{aligned}
\end{equation*}
where $\downarrow$ denotes the down-sampling operator.
Similarly, for $k=0,\ldots, N_{j-1}$ and $n=1,\ldots,r$,
\begin{equation*}
    \begin{aligned}
    w_{j-1, k}^{n}&=\sum_{\ell=0}^{\Lambda_{j}} \widehat{f_{\ell}} \overline{\widehat{\beta^{n}}\left(\frac{\lambda_{\ell}}{2^{j-1}}\right)} \sqrt{\omega_{j-1, k}} u_{\ell}\left(\boldsymbol{x}_{j-1, k}\right)\\
    &=\left(v_{j} *_{j}\left(b^{n}\right)^{*}\right)_{k} .
    \end{aligned}
\end{equation*}
Therefore, we have the following identity,
\begin{equation*}
    \begin{aligned}
    \widetilde{v}&:=\left(v_{j-1} \uparrow_{j}\right) *_{j} a+\sum_{n=1}^{r} w_{j-1}^{n} *_{j} b^{n}\\
    &=\left(\left(\left(v_{j} *_{j} a^{\star}\right) \downarrow_{j}\right) \uparrow_{j}\right) *_{j} a+\sum_{n=1}^{r}\left(v_{j} *_{j}\left(b^{n}\right)^{\star}\right) *_{j} b^{n}
    \end{aligned}
\end{equation*}
The multi-level decomposition and reconstruction algorithms are shown as Algorithms~\ref{algo:decomp.multi.level} and \ref{algo:rec.multi.level}. The Fourier transforms in the algorithms can be implemented by FFT. Thus, the fast multi-level needlet transform on $\mathbb{S}^2$ with $N$ the size of the input data has the computational complexity $\mathcal{O}(N\sqrt{\log N})$.
With needlet filters given, we can pre-compute the need coefficients and store data as frequency domain signals. Further decomposition into finer scale and reconstruction to obtain lower-level approximation information are also optional, depending on requirements of applications.

\section{Error Bound of Rotation Equivariance}\label{bound}
In order to reduce the numerical error caused by repeated forward and backward FFTs, and also to decrease the model complexity, we apply non-linear shrinkage function on the high-pass coefficients with an controllable parameter $\sigma$, which is an analogue to the noise level of the denoising model. Since the low-pass coefficients provide approximate information of the input signal, our model contains the good property of approximate rotation equivariance. According to the needlets theory, the rotation equivariance error brought about by using the shrinkage on the high passes is estimable. The error is defined as Eq.~(\ref{error_thm}), which has the convergence order $2^{-(J_0+1)s}$.
\begin{proof}
Define 
$$f_J = f_{J_0}^{(\textup{L})}+f_{J_0}^{(\textup{H}),J}=f_{J_0}^{(\textup{L})}+\sum_{j=J_0}^J\left<f,\psi_j\right>\psi_j$$ 
as the spherical needlet approximation. By \citet[Theorem 3.12]{wang2017fully}, for $f\in \mathbb{W}_p^s(\mathbb{S}^2)$ with $s>0$ and $J\geqslant 0$, we have $\left\| f-f_J\right\|\leqslant C_1 2^{-Js}$ and $\left\| f-f_{J_0}^{(\textup{L})}\right\|\leqslant C_2 2^{-J_0s}$, $C_1$ and $C_2$ are constants that depend only on $d,p,s,h,$ and filter smoothness $\kappa$. Therefore, 
\begin{equation}
    \begin{aligned}
    \left\|f_{J_0}^{(\textup{H}),J} \right\|^2 &= \sum_{\ell\leqslant 2^J}\left\|\widehat{f_\ell^{(\textup{H})}} \right\|^2 = \left\| f_J-f_{J_0}^{(\textup{L})}\right\|^2\\
    &\leqslant \left\|f-f_J \right\|^2 + \left\|f-f_{J_0}^{(\textup{L})} \right\|^2\\
    &\leqslant C_1 2^{-Js}+C_2 2^{-J_0s}\leqslant C2^{-J_0s},
    \end{aligned}
    \label{order}
\end{equation}
where $C$ is a sufficiently large constant depending on $d,p,s,h,\kappa, C_1$ and $C_2$. The Eq.~(\ref{order}) holds for all $J$. Then, Eq.~(\ref{error_thm}) satisfies the following inequalities,
\begin{equation*}
    \begin{aligned}
\textup{Error}&=\sum_{\ell=0}^B\left\| \textup{Shr}(\widehat{L_Rf\star \phi})_\ell^{(\textup{H})}-D^\ell(R)\textup{Shr}(\widehat{f\star\phi}_\ell^{(\textup{H})})\right\|^2\\
&=\sum_{\ell=0}^B\left\|\textup{Shr}(D^\ell(R)\widehat{f}_\ell^{(\textup{H})}\widehat{\phi}_\ell) -D^\ell(R)\textup{Shr}(\widehat{f}_\ell^{(\textup{H})}\widehat{\phi}_\ell)\right\|^2\\
&\leqslant\sum_{\ell=0}^B\left\| \textup{Shr}(D^\ell(R)\widehat{f}_\ell^{(\textup{H})}\widehat{\phi}_\ell)\right\|^2+\left\| D^\ell(R)\textup{Shr}(\widehat{f}_\ell^{(\textup{H})}\widehat{\phi}_\ell)\right\|^2\\
&\leqslant \left\| D^\ell(R)\widehat{f}_\ell^{(\textup{H})}\widehat{\phi}_\ell\right\|^2+\left\| D^\ell(R)\right\|^2\left\|\widehat{f}_\ell^{(\textup{H})}\widehat{\phi}_\ell\right\|^2\\
&\leqslant 2\sum_{\ell=0}^B\left\|\widehat{f}_\ell^{(\textup{H})}\widehat{\phi}_\ell \right\|^2.
    \end{aligned}
\end{equation*}
If the scale of the low-pass is $J_0$, then we obtain
$$\textup{Error}\leqslant 2\sum_{\ell\geqslant 2^{J_0+1}}\left\|\widehat{f}_\ell^{(\textup{H})}\widehat{\phi}_\ell \right\|^2\leqslant C_{\phi}\sum_{\ell\geqslant 2^{J_0+1}}\left\|\widehat{f}_\ell^{(\textup{H})} \right\|^2,$$
where $C_\phi$ is a constant depending on the filter $\phi$. By Parseval's identity and Eq.~(\ref{order}), 
$$\textup{Error}\leqslant C_{\phi}\sum_{\ell\geqslant 2^{J_0+1}}\left\|f_{J_0+1}^{(\textup{H})} \right\|^2\leqslant\Tilde{C_\phi}2^{-(J_0+1)s}.$$
\end{proof}

\section{Rotation Equivariance of Spectral Pooling}
Denote $[D^{0}(R), \cdots, D^{\ell-1}(R),D^{\ell}(R)]$ as $D(R)$, then
\begin{equation*}
    \begin{split}
        \textup{P}(\widehat{L_Rf})&=\textup{P}([D^{0}(R)\hat{f}_{0},\cdots,D^{\ell}(R)\hat{f}_{\ell}])\\
        &=[D^{0}(R)\hat{f}_{0},\cdots,D^{\ell/2}(R)\hat{f}_{\ell/2}]\\
        &=\textup{P}(\hat{f})\odot[D^{0}(R),\cdots,D^{\ell/2}(R)]\\
        &=\textup{P}(\hat{f})\odot\textup{P}(D(R))
    \end{split}
\label{Pooling_proof}
\end{equation*}
where $\odot$ denotes element-wise multiplication, and \textup{P($\cdot$)} denotes spectral pooling operator. Thus, the spectral pooling operator is equivariant, due to the $R$-related operator $\textup{P}(D(R))$.

\section{Ablation Study}
\subsection{Equivariance Error}
It can be proven that our $\mathbb{S}^2$-needlet convolution, $\mathrm{SO(3)}$-needlet convolution without shrinkage and spectral pooling are equivariant to $\mathrm{SO(3)}$ transforms for the continuous case. In our implementation, we utilize the polynomial-exact quadrature rule to sample the sphere to reduce the numerical error due to discretization. Table~\ref{tab:eq err} shows the rotation equivariance error of the modules in our framework. Experimental results verify that the errors are close to the machine error of floating points, except $\mathrm{SO(3)}$-needlet convolution with shrinkage filtering. We observe that the equivariance errors introduced by wavelet shrinkage with a small value of $\sigma$ (e.g., $\sigma = 0.001$) are 2e-4 and 5e-7 in   \textit{Single} and \textit{Double} floating-point format respectively, which are negligible. 
\begin{table}[t]
    \centering
    \setlength{\tabcolsep}{0.5mm}
\begin{tabular}{lcc}
\toprule
Operator& Error (\textit{Single}) & Error (\textit{Double}) \\ 
\midrule
\textsc{$\mathbb{S}^2$-Conv}& 2e-7& 7e-16 \\
\textsc{SO(3)-Conv}& 1e-7& 8e-16\\
\textsc{SO(3)+ReLU}& 1e-7& 8e-16\\
\textsc{SO(3)+Shrinkage} & 2e-4& 5e-7  \\
\textsc{Pooling}& 0 & 0\\
\bottomrule
\end{tabular}
\vspace{-0.2cm}
\caption{Equivariance Error Results. \textit{Single} denotes Single-precision floating-point format. \textit{Double} denotes Double-precision floating-point format. The values are calculated from the average of ten trials. ReLU function in \textsc{SO(3)+ReLU} is applied in the spatial domain, thus involving an FFT and inverse FFT. }\label{tab:eq err}\vspace{1mm}
\end{table}

\subsection{Sensitivity Analysis}
As $\sigma$ is a hyperparameter in shrinkage activation function, it is critical to know how this value affects our model equivariant property. Therefore, we takes different values of $\sigma$ ranging from 1e-7 to $1$ to see how the equivariant error changes and how much this signal is compressed. We use the SO(3) signal with the bandwidth $L = 128$ and send it to the SO(3)-needlet convolution layer with $J = 7$. As shown in Figure~\ref{equal error Ablation analysis}, when $\sigma$ is larger than $0.1$, the equivariance error is about $0.1$, which may lower the accuracy of our equivariant network. When the $\sigma$ is smaller than 1e-6, the equivariance error is approaching to the single-precision machine error. For compression rate, the shrinkage operation will cut off $20\%$ signal information when $\sigma$ is about $0.1$ and approaching to 0 when $\sigma$ is less than 1e-6.
\begin{figure}[th]
    \centering
    \includegraphics[width=0.44\linewidth]{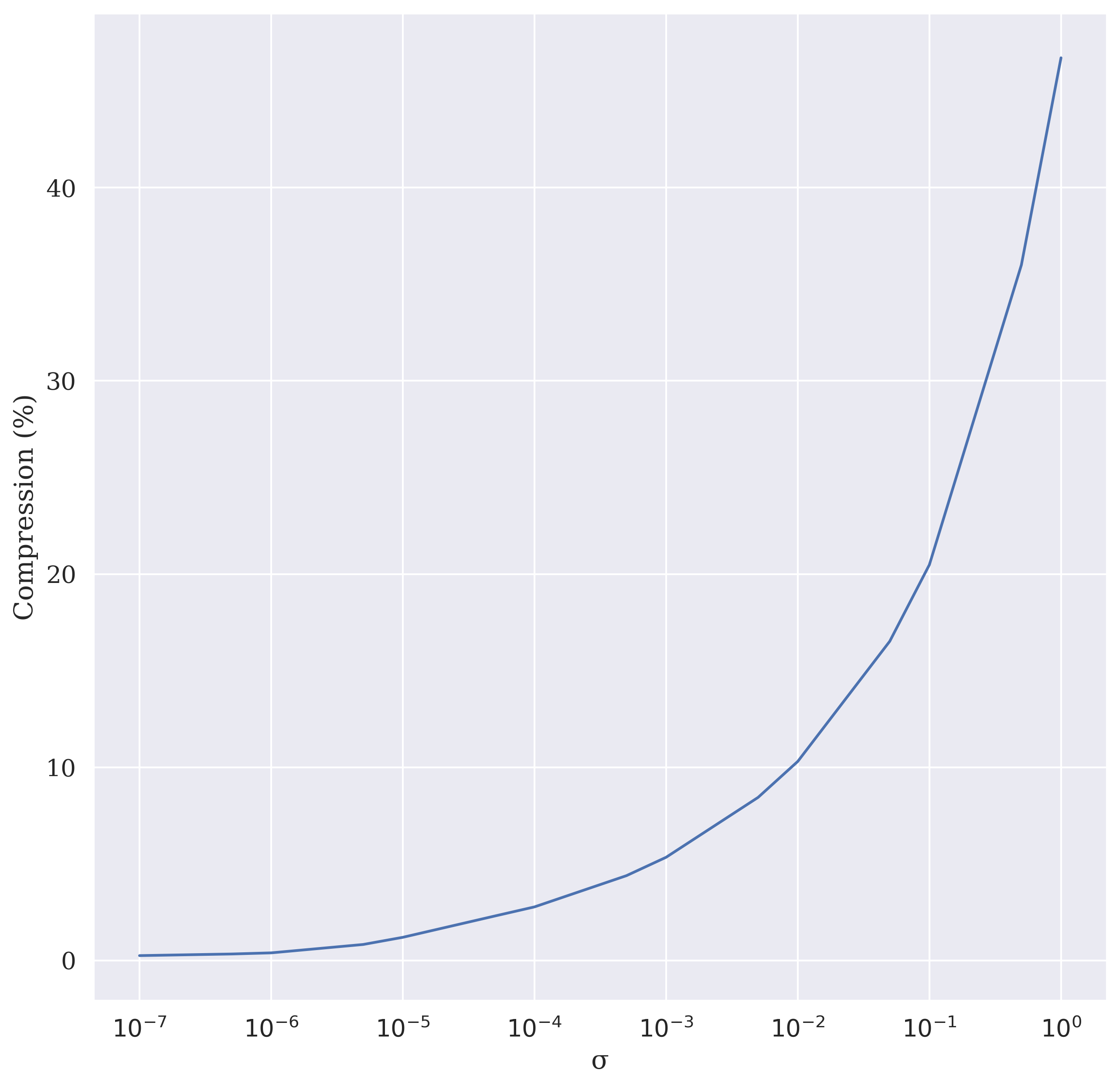}
    \includegraphics[width = 0.45\linewidth]{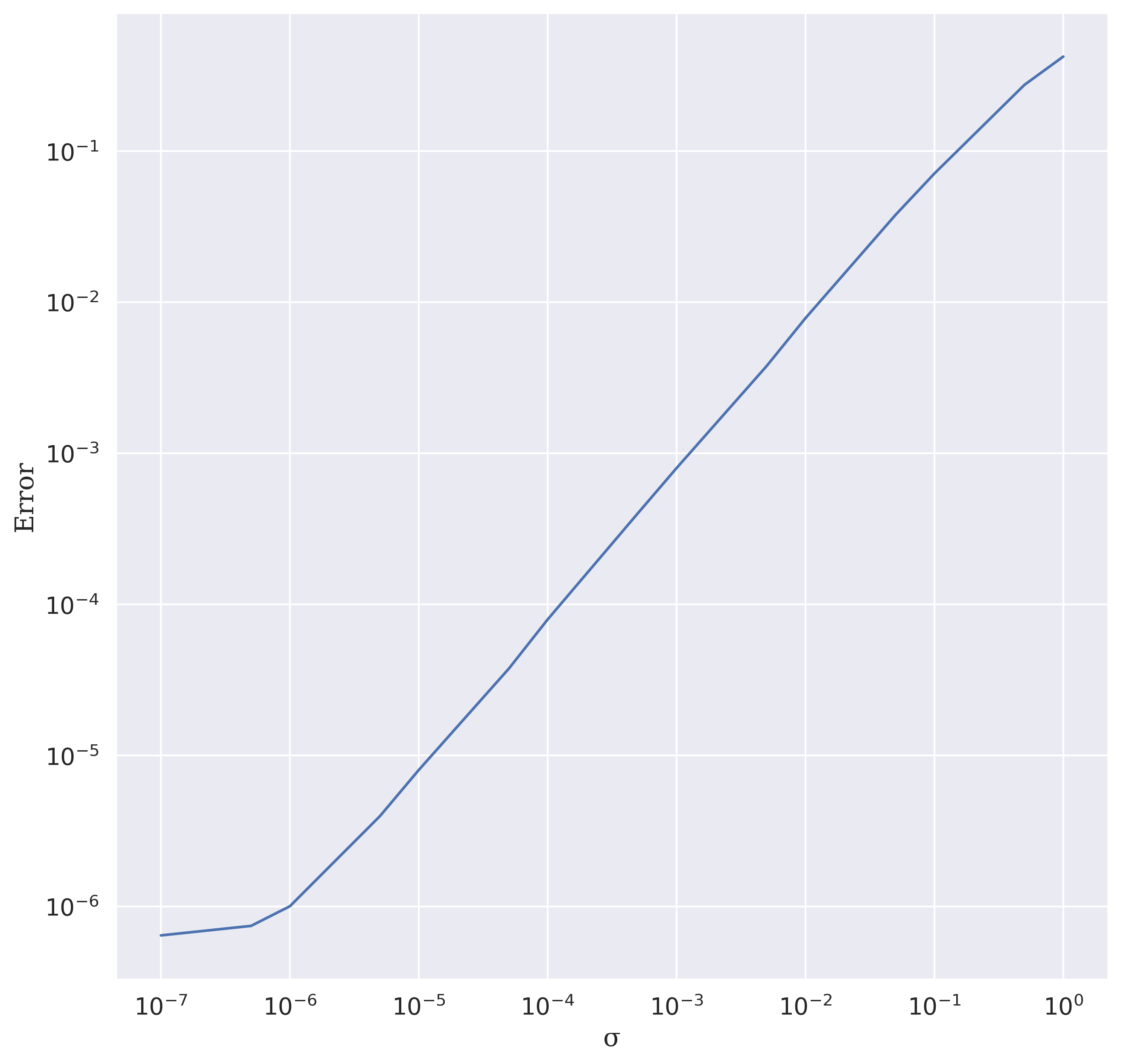}
    \caption{Sensitivity analysis for $\sigma$. The equivariance error is near machine error when $\sigma$ is less than 1e-6. The shrinkage activation function will nearly compress $20\%$ signal information with $\sigma=0.1$ and approximate the identity function as $\sigma$ is close to $0$.}
    \label{equal error Ablation analysis}
    \vspace{-0.5cm}
\end{figure}



\end{document}